\documentclass[a4paper,UKenglish]{lipics-v2018}

\usepackage{microtype}
\usepackage{times}
\usepackage{xcolor}
\usepackage{soul}
\usepackage[utf8]{inputenc}
\usepackage{hyperref}
\usepackage{amsmath}
\usepackage{amssymb}
\usepackage{amsthm}
\usepackage{stmaryrd}
\usepackage{graphicx}
\usepackage{subcaption}
\usepackage{booktabs}
\usepackage{verbatim}
\usepackage[ruled, linesnumbered, vlined]{algorithm2e} 

\newtheorem{property}{Property}
\newtheorem*{theorem*}{Theorem}
\newtheorem*{proposition*}{Proposition}
\newtheorem*{lemma*}{Lemma}

\def\ceplas{\textsf{\footnotesize AS}}
\def\ceplor{\textsf{\footnotesize OR}}
\def\ceplfilter{\textsf{\footnotesize FILTER}}
\def\ceplwin{\textsf{\footnotesize WINDOW}}
\def\trueb{\textsf{\footnotesize TRUE}}

\def\true3v{$\textsf{TRUE_{c}}$}
\def\false3v{$\textsf{FALSE_{c}}$}
\def\rg{$\mathit{RG}$}

\title{Symbolic Automata with Memory: a Computational Model for Complex Event Processing}

\titlerunning{Symbolic Automata with Memory}

\author{Elias Alevizos}{National Center for Scientific Research (NCSR) ``Demokritos'', Greece\\National Kapodistrian University of Athens, Greece}{alevizos.elias@iit.demokritos.gr}{}{}

\author{Alexander Artikis}{University of Piraeus, Greece\\National Center for Scientific Research (NCSR) ``Demokritos'', Greece}{a.artikis@unipi.gr}{}{}

\author{Georgios Paliouras}{National Center for Scientific Research (NCSR) ``Demokritos'', Greece}{paliourg@iit.demokritos.gr}{}{}

\authorrunning{E. Alevizos et. al.}

\Copyright{Elias Alevizos, Alexander Artikis and Georgios Paliouras}

\subjclass{
,
,
,
,
,
}


\category{}

\relatedversion{}

\supplement{}

\nolinenumbers 
\hideLIPIcs  

\begin{document}

\maketitle

\begin{abstract}
We propose an automaton model which is a combination of symbolic and register automata, i.e.,
we enrich symbolic automata with memory.
We call such automata Register Match Automata (RMA).
RMA extend the expressive power of symbolic automata,
by allowing formulas to be applied not only to the last element read from the input string,
but to multiple elements, stored in their registers. 
RMA also extend register automata, by allowing arbitrary formulas, besides equality predicates.
We study the closure properties of RMA under union, concatenation, Kleene$+$, complement and determinization and show that RMA,
contrary to symbolic automata,
are not determinizable when viewed as recognizers,
without taking the output of transitions into account.
However, when a window operator, a quintessential feature in Complex Event Processing, is used,
RMA are indeed determinizable even when viewed as recognizers.
We present detailed algorithms for constructing deterministic RMA from regular expressions extended with $n$-ary constraints.
We show how RMA can be used in Complex Event Processing in order to detect patterns upon streams of events,
using a framework that provides denotational and compositional semantics, 
and that allows for a systematic treatment of such automata.
\end{abstract}

\section{Introduction}

A Complex Event Processing (CEP) system takes as input a stream of events,
along with a set of patterns,
defining relations among the input events,
and detects instances of pattern satisfaction,
thus producing an output stream of complex events \cite{luckham2008power,cugola_processing_2012}.
Typically, an event has the structure of a tuple of values which might be numerical or categorical.
Since time is of critical importance for CEP,
a temporal formalism is used in order to define the patterns to be detected.
Such a pattern imposes temporal (and possibly atemporal) constraints on the input events,
which, if satisfied, lead to the detection of a complex event.
Atemporal constraints may be ``local'',
applying only to the last event read,
e.g., 
in streams from temperature sensors,
the constraint that the temperature of the last event is higher than some constant threshold.
Alternatively, they might involve multiple events of the pattern,
e.g., 
the constraint that the temperature of the last event is higher than that of the previous event.

Automata are of particular interest for the field of CEP,
because they provide a natural way of handling sequences.
As a result, 
the usual operators of regular expressions, concatenation, union and Kleene$+$,
have often been given an implicit temporal interpretation in CEP.
For example, 
the concatenation of two events is said to occur whenever the second event is read by an automaton after the first one,
i.e.,
whenever the timestamp of the second event is greater than the timestamp of the first 
(assuming the input events are temporally ordered).
On the other hand,
atemporal constraints are not easy to define using classical automata,
since they either work without memory or, 
even if they do include a memory structure,
e.g., as with push-down automata,
they can only work with a finite alphabet of input symbols.
For this reason,
the CEP community has proposed several extensions of classical automata.
These extended automata have the ability to store input events and later retrieve them in order to evaluate whether a constraint is satisfied \cite{demers_cayuga_2007,agrawal_efficient_2008,cugola_tesla_2010}.
They resemble both register automata \cite{kaminski_finite-memory_1994},
through their ability to store events,
and symbolic automata \cite{dantoni_power_2017},
through the use of predicates on their transitions.
They differ from symbolic automata in that predicates apply to multiple events, 
retrieved from the memory structure that holds previous events.
They differ from register automata in that predicates may be more complex than that of (in)equality.

One issue with these automata is that their properties have not been systematically investigated,
as is the case with models derived directly from the field of languages and automata.
See \cite{DBLP:journals/corr/abs-1709-05369} for a discussion about the weaknesses of automaton models in CEP. 
Moreover, they sometimes need to impose restrictions on the use of regular expression operators in a pattern, 
e.g., nesting of Kleene closure operators is not allowed.
A recently proposed formal framework for CEP attempts to address these issues \cite{DBLP:journals/corr/abs-1709-05369}.
Its advantage is that it provides a logic for CEP patterns, called CEPL, with simple denotational and compositional semantics, 
but without imposing severe restrictions on the use of operators.
A computational model is also proposed, through the so-called \textit{Match Automata} (MA), which may be conceived as variations of symbolic transducers \cite{dantoni_power_2017}. 
However, MA can only handle ``local'' constraints,
i.e.,
the formulas on their transitions are unary and thus are applied only to the last event read.
We propose an automaton model that is an extension of MA.
It has the ability to store events and its transitions have guards in the form of $n$-ary formulas.
These formulas may be applied both to the last event and to past events that have been stored. 
We call such automata \textit{Register Match Automata} (RMA).
RMA extend the expressive power of MA, symbolic automata and register automata, 
by allowing for more complex patterns to be defined and detected on a stream of events.
The contributions of the paper may be summarized as follows:
\begin{itemize}
	\item We present an algorithm for constructing a RMA from a regular expression with constraints in which events may be constrained through $n$-ary formulas, as a significant extension of the corresponding algorithms for symbolic automata and MA. 
	\item We prove that RMA are closed under union, concatenation, Kleene$+$ and determinization but not under complement.
	\item We show that RMA, when viewed as recognizers, are not determinizable.
	\item We show that patterns restricted through windowing, a common constraint in CEP, can be converted to a deterministic RMA, if the output of the transitions is not taken into account, i.e., if RMA are viewed as recognizers.
\end{itemize}
A selection of proofs and algorithms for the most important results may be found in the Appendix. 

\section{Related Work}

Because of their ability to naturally handle sequences of characters,
automata have been extensively adopted in CEP,
where they are adapted in order to handle streams composed of tuples.
Typical cases of CEP systems that employ automata are
the Chronicle Recognition System \cite{ghallab_chronicles_1996,dousson_chronicle_2007},
Cayuga \cite{demers_cayuga_2007},
TESLA \cite{cugola_tesla_2010} and
SASE \cite{agrawal_efficient_2008,zhang_complexity_2014}.
There also exist systems that do not employ automata as their computational model,
e.g., there are logic-based systems \cite{artikis2015event} or systems that use trees \cite{mei2009zstream},
but the standard operators of concatenation, union and Kleene$+$ are quite common and they may be considered as a reasonable set of core operators for CEP.
For a tutorial on CEP languages, see \cite{artikis2017complex},
and for a general review of CEP systems, see \cite{cugola_processing_2012}.
However, current CEP systems do not have the full expressive power of regular expressions,
e.g.,
SASE does not allow for nesting Kleene$+$ operators. 
Moreover, due to the various approaches
implementing the basic operators and extensions in their own way,
there is a lack of a common ground that could act as a basis for systematically understanding the properties of these automaton models.
The abundance of different CEP systems,
employing various computational models and using various formalisms 
has recently led to some attempts at providing a unifying framework 
\cite{DBLP:journals/corr/abs-1709-05369,DBLP:journals/corr/Halle17}.
Specifically, in \cite{DBLP:journals/corr/abs-1709-05369},
a set of core CEP operators is identified,
a formal framework is proposed that provides denotational semantics for CEP patterns, 
and a computational model is described,
through \emph{Match Automata} (MA), 
for capturing such patterns.

Outside the field of CEP,
research on automata has evolved towards various directions.
Besides the well-known push-down automata that can store elements from a finite set to a stack,
there have appeared other automaton models with memory,
such as register automata, 
pebble automata and 
data automata \cite{kaminski_finite-memory_1994,neven_finite_2004,bojanczyk2011two}.
For a review, see \cite{segoufin_automata_2006}.
Such models are especially useful when the input alphabet cannot be assumed to be finite,
as is often the case with CEP.
Register automata (initially called finite-memory automata) constitute one of the earliest such proposals \cite{kaminski_finite-memory_1994}.
At each transition,
a register automaton may choose to store its current input (more precisely, the current input's data payload)
to one of a finite set of registers.
A transition is followed if the current input complies with the contents of some register.
With register automata,
it is possible to recognize strings constructed from an infinite alphabet,
through the use of (in)equality comparisons among the data carried by the current input and the data stored in the registers.
However,
register automata do not always have nice closure properties,
e.g.,
they are not closed under determinization
(see \cite{libkin2015regular} for an extensive study of register automata).
Another model that is of interest for CEP is the symbolic automaton,
which allows CEP patterns to apply constraints on the attributes of events.
Automata that have predicates on their transitions were already proposed in \cite{noord_finite_2001}.
This initial idea has recently been expanded and more fully investigated in symbolic automata 
\cite{veanes_symbolic_2010,veanes_applications_2013,dantoni_power_2017}.
In this automaton model,
transitions are equipped with formulas constructed from a Boolean algebra.
A transition is followed
if its formula,
applied to the current input,
evaluates to true.
Contrary to register automata,
symbolic automata have nice closure properties,
but their formulas are unary and thus can only be applied to a single element from the input string.

This is the limitation that we address in this paper,
i.e.,
we propose an automaton model, called \textit{Register Match Automata} (RMA), 
whose transitions can apply $n$-ary formulas (with $n{>}1$) on multiple elements.
RMA are thus more expressive than symbolic automata (and Match Automata),
thus being suitable for practical CEP applications,
while, at the same time,
their properties can be systematically investigated,
as in standard automata theory.

\section{Grammar for Patterns with $n$-ary Formulas}
\label{section:cepl}

Before presenting RMA,
we first briefly present a high-level formalism for defining CEP patterns,
called ``CEP logic'' (CEPL), 
introduced in \cite{DBLP:journals/corr/abs-1709-05369}
(where a detailed exposition and examples may be found).

We first introduce an example from \cite{DBLP:journals/corr/abs-1709-05369} that will be used throughout the paper to provide intuition.
The example is that of a set of sensors taking temperature and humidity measurements,
monitoring an area for the possible eruption of fires.
A stream is a sequence of events,
where each event is a tuple of the form $(\mathit{type},\mathit{id},\mathit{value})$.
The first attribute ($\mathit{type}$) is the type of measurement: 
$H$ for humidity and $T$ for temperature.
The second one ($\mathit{id}$) is an integer identifier, unique for each sensor.
It has a finite set of possible values.
Finally, the third one ($\mathit{value}$) is the real-valued measurement from a possibly infinite set of values.
Table \ref{table:example_stream} shows an example of such a stream.
We assume that events are temporally ordered and their order is implicitly provided through the index.

\begin{table}[t]
\centering
\caption{Example stream.}
\begin{tabular}{cccccccc} 
\toprule
type & T & T & T & H & H & T & ... \\ 
\midrule
id & 1 & 1 & 2 & 1 & 1 & 2 & ... \\
\midrule
value & 22 & 24 & 32 & 70 & 68 & 33 & ... \\
\midrule
index & 0 & 1 & 2 & 3 & 4 & 5 & ... \\
\bottomrule
\end{tabular}
\label{table:example_stream}
\end{table}

The basic operators of CEPL's grammar are
the standard operators of regular expressions,
i.e.,
concatenation, union and Kleene$+$,
frequently referred to with the equivalent terms
\textit{sequence}, \textit{disjunction} and \textit{iteration} respectively.
The formal definition is as follows \cite{DBLP:journals/corr/abs-1709-05369}:
\begin{definition}[core--CEPL grammar]
\label{definition:core_cepl_grammar}
The core--CEPL grammar is defined as:
\begin{equation*}
\phi := R\ \ceplas\ x\ |\ \phi\ \ceplfilter\ f\ |\ \phi\ \ceplor\ \phi \ |\ \phi;\phi\ |\ \phi^{+}
\end{equation*}
where $R$ is a relation name, 
$x$ a variable, 
$f$ a selection formula,
``$;$'' denotes sequence, 
``$\ceplor$'' denotes disjunction
and ``$^{+}$'' denotes iteration.
\end{definition}
Intuitively,
$R$ refers to the type of an event 
(e.g., $T$ for temperature)
and variables $x$ are used in order to be able to refer to events involved in a pattern through the \ceplfilter\ constraints
(e.g., $T\ \ceplas\ x\ \ceplfilter\ x.\mathit{value} > 20$).
From now on,
we will use the term ``expression'' to refer to CEPL patterns defined as above
and the term ``formula'' to refer to the selection formulas $f$ in \ceplfilter\ expressions.
Note that extended versions of CEPL include more operators,
beyond the core ones presented above,
but these will not be treated in this paper.
We reserve such a treatment for future work.

Assume that $S=t_{0}t_{1}t_{2}\cdots$ is a stream of events/tuples
and $\phi$ a CEPL expression.
Our aim is to detect matches of $\phi$ in $S$.
A match $M$ is a set of natural numbers,
referring to indices in the stream.
If $M{=}\{i_{1},i_{2},\cdots\}$ is a match for $\phi$,
then the set of tuples referenced by $M$,
$S[M]{=}\{t_{i_{1}}, t_{i_{2}}, \cdots \}$ represents a complex event
(of type $\phi$).
Determining whether an arbitrary set of indices is a match for an expression
requires a definition for the semantics of CEPL expressions,
which may be found in 
\cite{DBLP:journals/corr/abs-1709-05369}.
There is one remark that is worth making at this point.
Let 
$\phi{:=}(T\ \ceplas\ x);(H\ \ceplas\ y)$
be a CEPL expression for our running example.
It aims at detecting pairs of events in the stream,
where the first is a temperature measurement and the second a humidity measurement.
Readers familiar with automata theory might expect that,
when applied to the stream of Table \ref{table:example_stream},
it would detect only $M{=}\{ 2, 3\}$ as a match.
However, 
in CEP,
such contiguous matches are not always the most interesting.
This is the reason why,
according to the CEPL semantics,
all the possible pairs of $T$ events followed by $H$ events are accepted as matches.
Specifically,
$\{ 0, 3 \}$, $\{ 0, 4 \}$, $\{ 1, 3 \}$, $\{ 1, 4 \}$, $\{ 2, 3 \}$, $\{ 2, 4 \}$
would all be matches.
There are ways to enforce a more ``classical behavior'' for CEPL expressions,
like accepting only contiguous matches,
but this requires the notion of \textit{selection strategies} \cite{DBLP:journals/corr/abs-1709-05369,zhang_complexity_2014}.
We only deal with the default ``behavior'' of CEPL expressions.
As another example,
let
\begin{equation}
\label{expression:t_seq_h_filter_eq_id}
\phi_{1} := (T\ \ceplas\ x);(H\ \ceplas\ y)\ \ceplfilter\ (x.\mathit{id} = y.\mathit{id})
\end{equation}
be a CEPL expression,
as previously,
but with the binary formula $x.\mathit{id}{=}y.\mathit{id}$ as an extra constraint.
The matches for this expression would be the same,
except for $\{ 2, 3 \}$, 
since events/tuples $t_{2}$ and $t_{3}$ have different sensor identifiers.

The semantics of CEPL requires the notions of valuations
(a valuation is a partial function $v: \mathbf{X} \rightharpoonup \mathbb{N}$,
mapping variables to indices, see \cite{DBLP:journals/corr/abs-1709-05369}) 
and may be informally given as follows:
The base case, $R\ \ceplas\ x$, 
is similar to the base case in classical automata.
We check whether the event is of type $R$,
i.e.,
if $M=\{i\}$, $v(x)=i$ for the valuation $v$ and the type of $t_{i}$ is $R$,
then $M$ is indeed a match.
For the case of expressions like $\phi\ \ceplfilter\ f(x,y,z,\cdots)$,
$M$ under $v$ must be a match of the sub-expression $\phi$.
In addition, 
the tuples associated with the variables $x,y,z,\cdots$ through $v$ must satisfy $f$,
i.e.,
$f(t_{v(x)},t_{v(y)},t_{v(z)},\cdots)=\trueb$.
If $\phi := \phi_{1}\ \ceplor\ \phi_{2}$,
$M$ must be a match either of $\phi_{1}$ or of $\phi_{2}$.
If $\phi := \phi_{1}\ ;\ \phi_{2}$,
then we must be able to split $M$ in two matches $M=M_{1} \cdot M_{2}$ 
($M_{2}$ follows $M_{1}$) 
so that 
$M_{1}$ is a match of $\phi_{1}$
and $M_{2}$ is a match of $\phi_{2}$.
Finally, 
for $\phi := \psi^{+}$,
we must be able to split $M$ in $n$ matches $M=M_{1} \cdot M_{2} \cdots M_{n}$
so that $M_{1}$ is a match of the sub-expression $\psi$ 
(under the initial valuation $v$)
and the subsequent matches $M_{i}$ are also matches of $\psi$ 
(under new valuations $v_{i}$).
The fact that $M$ is a match of $\phi$ over a stream $S$,
starting at index $i \in \mathbb{N}$,
and under the valuation $v$ is denoted by $M \in \llbracket \phi \rrbracket(S,i,v)$ \cite{DBLP:journals/corr/abs-1709-05369}.


Variables in CEPL expressions are useful for defining constraints in the form of formulas.
However, careless use of variables may lead to some counter-intuitive and undesired consequences.
The notions of \textit{well--formed} and \textit{safe} expressions 
deal with such cases \cite{DBLP:journals/corr/abs-1709-05369}.
For our purposes,
we need to impose some further constraints on the use of variables.
Our aim is to construct an automaton model that can capture CEPL expressions with $n$-ary formulas.
In addition, 
we would like to do so with automata that have a finite number of registers,
where each register is a memory slot that can store one event.
The reason for the requirement of bounded memory is that automata with unbounded memory have two main disadvantages:
they often have undesirable theoretical properties, 
e.g., push-down automata are not closed under determinization; 
and they are not a realistic option for CEP applications,
which always work with restricted resources.
Under the CEPL semantics though, 
it is not always possible to capture patterns with bounded memory.
This is the reason why we restrict our attention to a fragment of core--CEPL that can be evaluated with bounded memory. 
As an example of an expression requiring unbounded memory, consider the following:
\begin{equation}
\label{eq:t_filter_eq_id_plus_seq_h}
\phi_{3} := (T\ \ceplas\ x\ \ceplfilter\ x.id{=}y.id)^{+} ; (H\ \ceplas\ y)
\end{equation}
Although a bit counter-intuitive,
it is well--formed.
It captures a sequence of one or more $T$ events,
followed by a $H$ event 
and the \ceplfilter\ formula checks that all these events are from the same sensor.
$ M{=}\{ 0, 1, 3\} $ would be a match for this expression in our example.
However, 
if more $T$ events from the sensor with $\mathit{id}{=}1$ were present before the $H$ event,
then these should also constitute a match,
regardless of the number of these $T$ events.
An automaton trying to capture such a pattern would need to store all the $T$ events,
until it sees a $H$ event and can compare the $\mathit{id}$ of this $H$ event with the $\mathit{id}$ of every previous $T$ event.
Therefore, such an automaton would require unbounded memory.
Note that, 
for this simple example with the equality comparison,
an automaton could be built that stores only the first $T$ event and then checks this event's $\mathit{id}$ with the $\mathit{id}$ of every new event.
In the general case and for more complex constraints though,
e.g.,
an inequality comparison,
all $T$ events would have to be stored.


We exclude such cases by focusing on the so-called \emph{bounded} expressions,
which are a specific case of \emph{well-formed} expressions.
Bounded expressions are formally defined as follows 
(see \cite{DBLP:journals/corr/abs-1709-05369} for a definition of $\mathit{bound}(\phi)$):
\begin{definition}[Bounded expression]
\label{definition:bounded_expressions}
A core--CEPL expression $\phi$ is bounded if it is well-formed and one of the following conditions hold:
\begin{itemize}
	\item $\phi := R\ \ceplas\ x$.
	\item $\phi := \psi\ \ceplfilter\ f$ and $\forall x \in var(f)$, we have that $x \in \mathit{bound}(\psi)$.
	\item $\phi := \phi_{1}\ \ceplor\ \phi_{2} \ |\ \phi_{1};\phi_{2}\ |\ \psi^{+}$ and all sub-expressions of $\phi$ are bounded. Moreover, if $\phi := \phi_{1};\phi_{2}$, then $var(\phi_{1})\cap var(\phi_{2})=\emptyset$.
\end{itemize}
\end{definition}

In other words,
for $\psi\ \ceplfilter\ f$,
variables in $f$ must be defined inside $\psi$ and not in a wider scope.
Additionally, if a variable is defined in a disjunct of an \ceplor\ operator,
then it must be defined in every other disjunct of this operator.
Variables defined inside a $^{+}$ operator are also not allowed to be used outside this operator and vice versa.
Finally, variables are not to be shared among sub-expressions of $;$ operators.
According to this definition then,
Expression \eqref{eq:t_filter_eq_id_plus_seq_h} is well-formed,
but not bounded,
since variable $y$ in $(T\ \ceplas\ x\ \ceplfilter\ x.id{=}y.id)$ does not belong to 
$\mathit{bound}(T\ \ceplas\ x)$. 
Note that this definition does not exclude nesting of regular expression operators.
For example, consider the following expression:
\begin{equation*}
\begin{aligned}
  ( &   (T\ \ceplas\ x_{1}\ ;\ T\ \ceplas\ x_{2}\ \ceplfilter\ x_{1}.\mathit{value} = x_{2}.\mathit{value}) ; \\
    & (H\ \ceplas\ x_{3}\ ;\ H\ \ceplas\ x_{4}\ \ceplfilter\ x_{3}.\mathit{value} = x_{4}.\mathit{value})^{+}  )^{+}
\end{aligned}
\end{equation*}
It has nested Kleene$+$ operators but is still bounded,
since variables are not used outside the scope of the Kleene$+$ operators where they are defined.

\section{Register Match Automata}
\label{section:cepl_model:rmas}

In order to capture bounded core--CEPL expressions, 
we propose Register Match Automata (RMA), an automaton model equipped  with memory,
as an extension of MA introduced in \cite{DBLP:journals/corr/abs-1709-05369}.
The basic idea is the following.
We add a set of registers \rg\ to an automaton in order to be able to store events from the stream
that will be used later in $n$-ary formulas. 
Each register can store at most one event.
In order to evaluate whether to follow a transition or not,
each transition is equipped with a guard, in the form of a formula.
If the formula evaluates to \trueb, then the transition is followed.
Since a formula might be $n$-ary, with $n{>}1$,
the values passed to its arguments during evaluation may be either the current event
or the contents of some registers,
i.e.,
some past events.
In other words, the transition is also equipped with a \textit{register selection},
i.e., a tuple of registers.
Before evaluation, the automaton reads the contents of those registers,
passes them as arguments to the formula 
and the formula is evaluated.
Additionally, if, during a run of the automaton, a transition is followed,
then the transition has the option to write the event that triggered it
to some of the automaton's registers.
These are called its \textit{write registers},
i.e.,
the registers whose contents may be changed by the transition.
Finally, each transition, when followed,
produces an output,
either $\circ$,
denoting that the event is not part of the match for the pattern that the RMA tries to capture,
or $\bullet$,
denoting that the event is part of the match.
We also allow for $\epsilon$-transitions, as in classical automata,
i.e., 
transitions that are followed without consuming any events and without altering the contents of the registers.

We now formally define RMA.
To aid understanding,
we present three separate definitions:
one for the automaton itself,
one for its transitions
and one for its configurations.

\begin{definition}[Register Match Automaton]
\label{definition:rmas:rma}
A RMA is a tuple ($Q$, $Q^{s}$, $Q^{f}$, \rg, $\Delta$)
where 
$Q$ is a finite set of states, 
$Q^{s}\subseteq Q$ the set of start states, 
$Q^{f}\subseteq Q$ the set of final states, 
\rg\ a finite set of registers
and $\Delta$ the set of transitions (see Definition \ref{definition:rmas:transition}).
When we have a single start state,
we denote it by $q^{s}$.
\end{definition}

For the definition of transitions,
we need the notion of a $\gamma$ function representing the contents of the registers, i.e., 
$\gamma: \mathit{RG} \cup \{\sim\} \rightarrow \mathit{tuples}(\mathcal{R})$.
The domain of $\gamma$ also contains $\sim$, representing the current event,
i.e., $\gamma(\sim)$ returns the last event consumed from the stream.

\begin{definition}[Transition of RMA]
\label{definition:rmas:transition}
A transition $\delta \in \Delta$ is a tuple $(q,f,rs,p,R,o)$, 
also written as $(q,f,rs)\rightarrow(p,R,o)$,
where
$q,p \in Q$, 
$f$ is a selection formula (as defined in \cite{DBLP:journals/corr/abs-1709-05369}),
$\mathit{rs}=(r_{1},\cdots,r_{n})$ the register selection,
	where $r_{i} \in \mathit{RG} \cup \{\sim\}$ ,
$R \in 2^{\mathit{RG}}$ the write registers and
$o \in \{\circ,\bullet\}$ is the set of outputs.
We say that a transition applies iff
$\delta{=}\epsilon$ and no event is consumed, 
or $f(\gamma(r_{1}),\cdots,\gamma(r_{n})){=}\trueb$ upon consuming an event.
\end{definition}

We will use the dot notation to refer to elements of tuples,
e.g., if $A$ is a RMA, then $A.Q$ is the set of its states.
For a transition $\delta$,
we will also use the notation $\delta.\mathit{source}$ and $\delta.\mathit{target}$ to refer to its source and target state respectively.
We will also sometimes write $\gamma(\mathit{rs})$ as shorthand notation for $(\gamma(r_{1}),\cdots,\gamma(r_{n}))$.

\begin{figure}[t]
\begin{centering}
\includegraphics[width=0.6\linewidth]{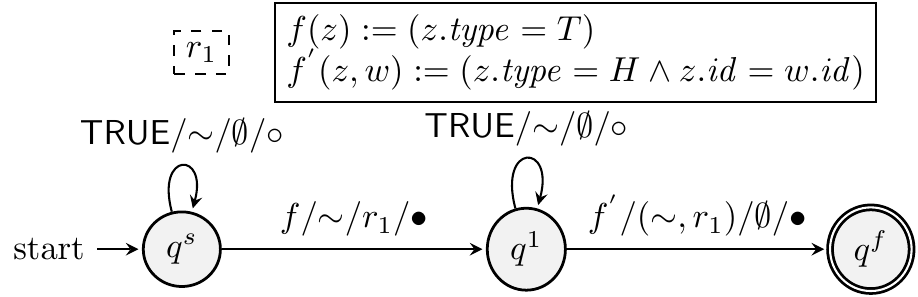}
\caption{RMA corresponding to Expression \eqref{expression:t_seq_h_filter_eq_id}.}
\label{fig:example1}
\end{centering}
\end{figure}

As an example,
consider the RMA of Fig. \ref{fig:example1}.
Each transition is represented as $f/\mathit{rs}/R/o$,
where $f$ is its formula,
$\mathit{rs}$ its register selection,
$R$ its write registers
and $o$ its output.
The formulas of the transitions are presented in a separate box,
above the RMA.
Note that the arguments of the formulas do not correspond to any variables of any CEPL expression,
but to registers, through the register selection
(we use $z$ and $w$ as arguments to avoid confusion with the variables of CEPL expressions).
Take the transition from $q^{s}$ to $q^{1}$ as an example.
It takes the last event consumed from the stream ($\sim$) and
passes it as argument to the unary formula $f$.
If $f$ evaluates to \trueb,
it writes this last event to register $r_{1}$,
displayed inside a dashed square in Fig. \ref{fig:example1}, 
and outputs $\bullet$.
On the other hand, 
the transition from $q^{1}$ to $q^{f}$ uses both the current event and the event stored in $r_{1}$ ($(\sim,r_{1})$) and passes them to the binary formula $f^{'}$.
Finally, the formula \trueb\ (for example, in the self-loop of $q^{s}$) is a unary predicate that always evaluates to \trueb.
The RMA of Fig. \ref{fig:example1} captures Expression \eqref{expression:t_seq_h_filter_eq_id}.

Note that there is a subtle issue with respect to how formulas are evaluated.
The definition about when a transition applies,
as it is,
does not take into account cases where the contents of some register(s) in a register selection are empty.
In such cases, 
it would not be possible to evaluate a formula
(or we would need a 3-valued algebra, 
like Kleene's or Lukasiewicz's; 
see \cite{bergmann2008introduction} for an introduction to many-valued logics).
For our purposes,
it is sufficient to require that all registers in a register selection are not empty
whenever a formula is evaluated
(they can be empty before any evaluation).
There is a structural property of RMA,
in the sense that it depends only on the structure of the RMA and is independent of the stream, that can satisfy this requirement.
We require that, for a RMA $A$, for every state $q$, if $r$ is a register in one of the register selections of the outgoing transitions of $q$, then $r$ must appear in every trail to $q$.
A trail is a sequence of successive transitions (the target of every transition must be the source of the next transition) starting from the start state,
without any state re-visits.
A walk is similarly defined, but allows for state re-visits. 
We say that a register $r$ appears in a trail if there exists at least one transition $\delta$ in the trail such that $r \in \delta.R$.
%

We can describe formally the rules for the behavior of a RMA through the notion of configuration:
\begin{definition}[Configuration of RMA]
\label{definition:rmas:configuration}
Assume a stream of events
$S=t_{0},t_{1},\cdots,t_{i},t_{i+1},\cdots$ 
and a RMA $A$ consuming $S$.
A configuration of $A$ is a triple
$c=[i,q,\gamma]$, 
where
$i$ is the index of the next event to be consumed,
$q$ is the current state of $A$ and
$\gamma$ the current contents of $A$'s registers.
We say that $c'=[i',q',\gamma']$ is a \emph{successor} of $c$ iff the following hold:
\begin{itemize}
	\item $\mathit{\exists\ \delta: (q,f,rs) \rightarrow (q',R,o)}$ applies.
	\item $i{=}i'$ if $\delta{=}\epsilon$. Otherwise $i'{=}i+1$.
	\item $\gamma' {=} \gamma$ if $\delta{=}\epsilon$ or $R{=}\emptyset$. Otherwise, 
	$\forall r \notin R, \gamma'(r){=}\gamma(r)$ and $\forall r \in R, \gamma'(r){=}t_{i}$. 
\end{itemize}
\end{definition}

For the initial configuration $c^{s}$,
before consuming any events,
we have that $c^{s}.q \in Q^{s}$ and, for each $r \in \mathit{RG}$, $c^{s}.\gamma(r){=}\sharp$,
where $\sharp$ denotes the contents of an empty register,
i.e.,
the initial state is one of the start states and all registers are empty.
Transitions from the start state cannot reference any registers in their register selection,
but only $\sim$. 
Hence, they are always unary.
In order to move to a successor configuration,
we need a transition whose formula evaluates to \trueb ,
applied to $\sim$, if it is unary, 
or to $\sim$ and the contents of its register selection, if it is $n$-ary.
If this is the case, 
we move one position ahead in the stream and update the contents of this transition's write registers,
if any, 
with the event that was read. 
If the transition is an $\epsilon$-transition, 
we do not move the stream pointer and do not update the registers,
but only move to the next state.
We denote a succession by $[i,q,\gamma] \rightarrow [i',q',\gamma']$,
or $[i,q,\gamma] \overset{\delta/o}{\rightarrow} [i',q',\gamma']$ if we need to refer to the transition and its output.

The actual behavior of a RMA upon reading a stream is captured by the notion of the run:
\begin{definition}[Run of RMA over stream]
A run $\varrho$ of a RMA $A$ over a stream $S$ is a sequence of successor configurations
$[0,q_{s},\gamma_{0}] \overset{\delta_{0}/o_{0}}{\rightarrow} [1,q_{1},\gamma_{1}] \overset{\delta_{1}/o_{1}}{\rightarrow} \cdots \overset{\delta_{n-1}/o_{n-1}}{\rightarrow} [n,q_{n},\gamma_{n}]$.
A run is called accepting iff $q_{n} \in Q^{f}$ and $o_{n-1}=\bullet$.
\end{definition}

A run of the RMA of Fig. \ref{fig:example1}, 
while consuming the first four events from the stream of Table \ref{table:example_stream}, 
is the following:
\begin{equation}
\label{run:example}
[0,q^{s},\sharp] \overset{\delta^{s,s}/\circ}{\rightarrow} [1,q^{s},\sharp] \overset{\delta^{s,1}/\bullet}{\rightarrow} [2,q^{1},(T,1,24)] \overset{\delta^{1,1}/\circ}{\rightarrow}
[3,q^{1},(T,1,24)] \overset{\delta^{1,f}/\bullet}{\rightarrow} [4,q^{f},(T,1,24)]
\end{equation}
Transition superscripts refer to states of the RMA,
e.g.,
$\delta^{s,s}$ is the transition from the start state to itself,
$\delta^{s,1}$ is the transition from the start state to $q^{1}$, etc.
Run \eqref{run:example} is not the only run,
since the RMA could have followed other transitions with the same input,
e.g.,
moving directly from $q^{s}$ to $q^{1}$.

The set of all runs over a stream $S$ that $A$ can follow is denoted by $Run_{n}(A,S)$
and the set of all accepting runs by $Run_{n}^{f}(A,S)$.
If $\varrho$ is a run of a RMA $A$ over a stream $S$ of length $n$,
by $\mathit{match}(\varrho)$ we denote all the indices in the stream that were ``marked'' by the run,
i.e.,
$\mathit{match}(\varrho){=}\{ i \in [0,n]: o_{i}{=}\bullet\}$.
For the example of Run \eqref{run:example},
we see that this run outputs a $\bullet$ after consuming $t_{1}$ and $t_{3}$.
Therefore, $match(\varrho){=}\{ 1, 3\}$.
We can also see that there exists another accepting run $\varrho^{'}$ for which $match(\varrho^{'}){=}\{ 0, 3\}$.
These are then the matches of this RMA after consuming the first four events of the example stream.
We formally define the matches produced by a RMA as follows,
similarly to the definition of matches of MA \cite{DBLP:journals/corr/abs-1709-05369}:
\begin{definition}[Matches of RMA]
The set of matches of a RMA $A$ over a stream $S$ at index $n$ is:
$\llbracket A \rrbracket_{n}(S) = \{ \mathit{match}(\varrho): \varrho \in Run_{n}^{f}(A,S)\}$.
The set of matches of a RMA $A$ over a stream $S$ is: $\llbracket A \rrbracket(S) = \bigcup\limits_{n}{ \llbracket A \rrbracket_{n}(S) }$.
\end{definition}

\section{Translating Expressions to Register Match Automata}
\label{section:cepl_model:nrmas}

We now show how,
for each bounded, core--CEPL expression with $n$-ary formulas,
we can construct an equivalent RMA.
Equivalence between an expression $\phi$ and a RMA $A_{\phi}$ means that a set of stream indices $M$ is a match of $\phi$ over a stream $S$ iff $M$ is a match of $A_{\phi}$ over $S$ or, more formally, $M \in \llbracket A_{\phi} \rrbracket(S_{i}) \Leftrightarrow \exists v: M \in \llbracket \phi \rrbracket(S,i,v)$.
\begin{theorem}
\label{theorem:cepl2nrmae}
For every bounded, core--CEPL expression (with $n$-ary formulas) 
there exists an equivalent RMA.
\end{theorem}


\begin{proof}[Proof and algorithm sketch]
The complete RMA construction algorithm and the full proof for the case of $n$-ary formulas and a single direction may be found in the Appendix.
Here, we first present an example, 
to give the intuition, 
and then present the outline of one direction of the proof.
Let 
\begin{equation}
\label{expression:cepl2rma_example}
\begin{aligned}
\phi_{4} := & 	(T\ \ceplas\ x\ \ceplfilter\ x.\mathit{value}<-40\ \ceplor\ T\ \ceplas\ x\ \ceplfilter\ x.\mathit{value}>50)\ ; \\
		& 	(T\ \ceplas\ y)\ \ceplfilter\ y.\mathit{id} = x.\mathit{id}
\end{aligned}
\end{equation}
be a bounded, core--CEPL expression.
With this expression,
we want to monitor sensors for possible failures.
We want to detect cases where a sensor records temperatures outside some range of values
($x$)
and continues to transmit measurements ($y$),
so that we are alerted to the fact that measurement $y$ might not be trustworthy.
The last $\ceplfilter$ condition is a binary formula,
applied to both $y$ and $x$.
Fig. \ref{fig:cepl2rma:example} shows the process for constructing
the RMA which is equivalent to Expression \eqref{expression:cepl2rma_example}.

\begin{figure}[t]
\begin{centering}
\includegraphics[width=0.8\linewidth]{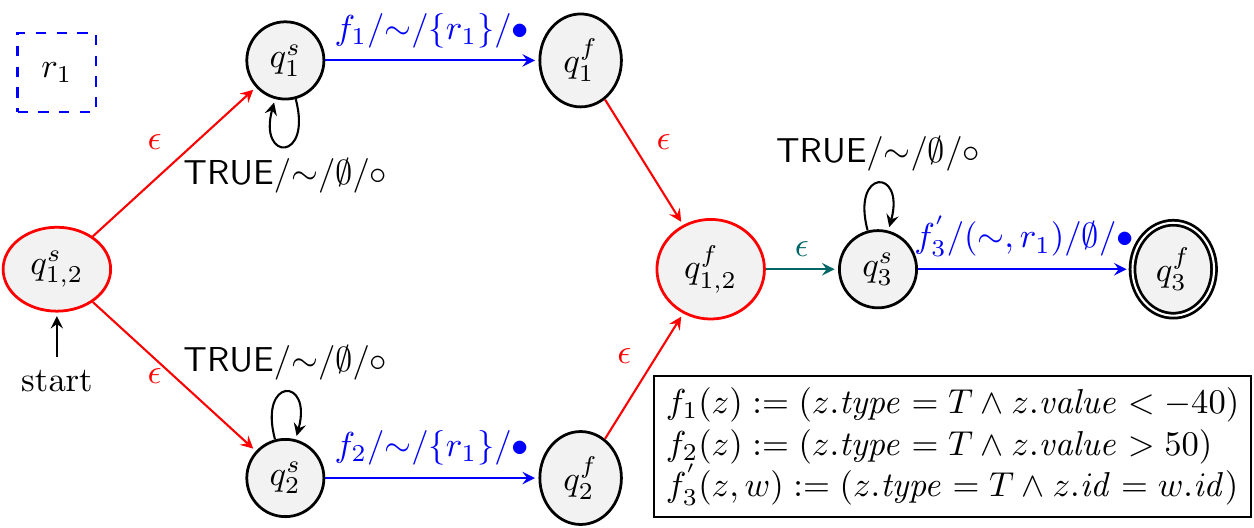}
\caption{Constructing RMA from the CEPL Expression \eqref{expression:cepl2rma_example}.}
\label{fig:cepl2rma:example}
\end{centering}
\end{figure}

\begin{algorithm}
\caption{CEPL to RMA for $n$-ary filter (simplified)}
\label{algorithm:cepl2rma:nary_simplified}
\KwIn{CEPL expression: $\phi = \phi^{'} \ceplfilter\ f(x_{1},\cdots,x_{n})$}
\KwOut{RMA $A_{\phi}$ equivalent to $\phi$}
$A_{\phi^{'}} \leftarrow \mathit{ConstructRMA}(\phi^{'})$\; 
\ForEach{transition $\delta$ of $A_{\phi^{'}}$}{
	\If{$x_{i}$ appears in $\delta$ and all other $x_{j}$ appear in all trails before $\delta$}{ \label{line:cepl2rma:nary_simplified:isvalid}
		add $f$ as a conjunct to the formula of $\delta$\; \label{line:cepl2rma:nary_simplified:append_f}
		\ForEach{transition $\delta_{j}$ before $\delta$}{
			get the variable $x_{j}$ associated with $\delta_{j}$\;
			\uIf{no register is associated with $x_{j}$}{
				create a new register $r_{j}$ associated with $x_{j}$\;
				make $\delta_{j}$ write to $r_{j}$\;
			}
			\Else {
				get register $r_{j}$ associated with $x_{j}$\;
			}
			add $r_{j}$ to the register selection of $\delta$\;
			
		}			
	}
}
return $A_{\phi^{'}}$\; 
\end{algorithm}

The algorithm is compositional,
starting from the base case $\phi {:=} R\ \ceplas\ x\ \ceplfilter\ f$.
The base case and the three regular expression operators (sequence, disjunction, iteration) are handled in a manner almost identical as for Match Automata,
with the exception of the sequence operator ($\phi=\phi_{1};\phi_{2}$),
where some simplifications are made due to the fact that expressions are bounded ($\mathit{var}(\phi_{1})\cap \mathit{var}(\phi_{2})=\emptyset$).
In this proof sketch, we focus on expressions with $n$-ary formulas, 
like $\phi = \phi^{'} \ceplfilter\ f(x_{1},\cdots,x_{n})$.

We first start by constructing the RMA for the base case expressions.
For the example of Fig. \ref{fig:cepl2rma:example},
there are three basic sub-expressions and three basic automata are constructed:
from $q^{s}_{1}$ to $q^{f}_{1}$,
from $q^{s}_{2}$ to $q^{f}_{2}$ and
from $q^{s}_{3}$ to $q^{f}_{3}$.
The first two are associated with variable $x$ of $\phi_{4}$ and the third with $y$.
To the corresponding transitions,
we add the relevant \emph{unary} formulas,
e.g.,
we add $f_{1}(z){:=} (z.\mathit{type}{=}T {\wedge} z.\mathit{value}{<}-40)$ to $q^{s}_{1} {\rightarrow} q^{f}_{1}$.
At this stage, 
since all formulas are unary,
we have no registers.
The \ceplor\ operator is handled by joining the RMA of the disjuncts through new states and $\epsilon$-transitions 
(see the red states and transitions in Fig. \ref{fig:cepl2rma:example}).
The ``$;$'' operator is handled by connecting the RMA of its sub-expressions through an $\epsilon$-transition,
without adding any new states (see the green transition).
Iteration,
not applicable in this example,
is handled by joining the final state of the original automaton to its start state through an $\epsilon$-transition.

Finally, for expressions with an $n$-ary formula 
we do not add any states or transitions.
We only modify existing transitions and possibly add registers,
as per Algorithm \ref{algorithm:cepl2rma:nary_simplified}
(this is a simplified version of Algorithm \ref{algorithm:cepl2rma:nary} in the Appendix).
For the example of Expression \eqref{expression:cepl2rma_example},
this new formula is $y.\mathit{id}{=}x.\mathit{id}$ and the transitions that are modified are shown in blue in Fig. \ref{fig:cepl2rma:example}.
First, we locate the transition(s) where the new formula should be added.
It must be a transition associated with one of the variables of the formula, which, in our example, means either with $x$ or $y$.
But the $x$-associated $q^{s}_{1} {\rightarrow} q^{f}_{1}$ and $q^{s}_{2} {\rightarrow} q^{f}_{2}$ should not be chosen,
since they are located before the $y$-associated $q^{s}_{3} {\rightarrow} q^{f}_{3}$,
if we view the RMA as a graph.
Thus, in a run, 
upon reaching either $q^{s}_{1} {\rightarrow} q^{f}_{1}$ or $q^{s}_{2} {\rightarrow} q^{f}_{2}$, 
the RMA won't have all the arguments necessary for applying the formula.
On the contrary,
the formula must be added to $q^{s}_{3} {\rightarrow} q^{f}_{3}$,
since,
at this transition, 
the RMA will have gone through one of the $x$-associated transitions
and seen an $x$-associated event.
By this analysis,
we can also conclude that events triggering $q^{s}_{1} {\rightarrow} q^{f}_{1}$ and $q^{s}_{2} {\rightarrow} q^{f}_{2}$
must be stored,
so that they can be retrieved when the RMA reaches $q^{s}_{3} {\rightarrow} q^{f}_{3}$.
Therefore, we add a register ($r_{1}$) and make them write to it.
Since these two transitions are in different paths of the same \ceplor\ operator and both refer to a common variable ($x$),
we add only a single register.
We then return back to $q^{s}_{3} {\rightarrow} q^{f}_{3}$ in order to update its formula.
Initially, its unary formula was $f_{3}(z){:=} (z.\mathit{type}{=}T)$.
We now add $r_{1}$ to its register selection and append the binary constraint $y.\mathit{id}{=}x.\mathit{id}$ as a conjunct,
thus resulting in $f^{'}_{3}(z,w) {:=} (z.\mathit{type}{=}T {\wedge} z.\mathit{id} {=} w.\mathit{id})$.

We provide a proof sketch for the case of $n$-ary formulas and for a single direction. 
We show how 
$M \in \llbracket A_{\phi} \rrbracket(S_{i}) \Rightarrow \exists v: M \in \llbracket \phi \rrbracket(S,i,v)$ 
is proven when 
$\phi = \phi^{'} \ceplfilter\ f(x_{1},\cdots,x_{n})$.
First, note that the proof is inductive,
with the induction hypothesis being that what we want to prove holds for the sub-expression $\phi^{'}$,
i.e., 
$M \in \llbracket A_{\phi^{'}} \rrbracket(S_{i}) \Leftrightarrow \exists v^{'}: M \in \llbracket \phi^{'} \rrbracket(S,i,v^{'})$.
We then prove the fact that
$M \in \llbracket A_{\phi} \rrbracket(S_{i}) \Rightarrow M \in \llbracket A_{\phi^{'}} \rrbracket(S_{i})$,
i.e., if $M$ is a match of $A_{\phi}$ then it should also be a match of $A_{\phi^{'}}$,
since $A_{\phi}$ has more constraints on some of its transitions.
We have thus proven the left-hand side of the induction hypothesis.
As a result, we can conclude that its right-hand side also holds,
i.e.,
$\exists v^{'}: M \in \llbracket \phi^{'} \rrbracket(S,i,v^{'})$.
Our goal is to find a valuation $v$ such that 
$M \in \llbracket \phi \rrbracket(S,i,v)$.
We can try the valuation $v^{'}$ that we just found for the sub-expression $\phi^{'}$.
We can show that $v^{'}$ is indeed a valuation for $\phi$ as well.
As per the definition of the CEPL semantics \cite{DBLP:journals/corr/abs-1709-05369},
to do so,
we need to prove two facts:
that $M \in \llbracket \phi^{'} \rrbracket(S,i,v^{'})$, which has just been proven;
and that $v^{'}_{s} \vDash f$, 
i.e., 
that $f$ evaluates to \trueb\ when its arguments are the tuples referenced by $v^{'}$.   
We can indeed prove the second fact as well,
by taking advantage of the fact that $M$ is produced by an accepting run of $A_{\phi}$. 
This run must have gone through a transition where $f$ was a conjunct and thus
$f$ does evaluate to \trueb.
\end{proof}

Note that the inverse direction of Theorem \ref{theorem:cepl2nrmae} is not necessarily true.
RMA are more powerful than bounded, core--CEPL expressions.
There are expressions which are not bounded but could be captured by RMA.
$(T\ \ceplas\ x) ; (H\ \ceplas\ y\ \ceplfilter\ y.id{=}x.id)^{+}$ is such an example.
An automaton for this expression would not need to store any $H$ events.
It would suffice for it to just compare the $id$ of every newly arriving $H$ event with the $id$ of the stored (and single) $T$ event.
A complete investigation of the exact class of CEPL expressions that can be captured with bounded memory is reserved for the future.
The construction algorithm for RMA uses $\epsilon$-transitions. 
As expected, it can be shown that such $\epsilon$-transitions can be removed from a RMA. 
The proof and the elimination algorithm are standard and are omitted.

We now study the closure properties of RMA under union, concatenation, Kleene$+$, complement and determinization.
We first provide the definition for deterministic RMA.
Informally, a RMA is said to be deterministic if it has a single start state and, at any time, with the same input event, it can follow no more than one transition with the same output.
The formal definition is as follows:
\begin{definition}[Deterministic RMA (DRMA)]
\label{definition:deterministic_rma}
Let $A$ be a RMA and $q$ a state of $A$. 
We say that $A$ is deterministic if for all transitions $(q,f_{1},rs_{1})\rightarrow(p_{1},R_{1},o)$,
$(q,f_{2},rs_{2})\rightarrow(p_{2},R_{2},o)$ 
(transitions from the same state $q$ with the same output $o$)
$f_{1}$ and $f_{2}$ are mutually exclusive,
i.e.,
at most one can evaluate to \trueb.
\end{definition}

This notion of determinism is similar to that used for MA in \cite{DBLP:journals/corr/abs-1709-05369}.
According to this notion, the RMA of Fig. \ref{fig:example1} is deterministic, 
since the two transitions from the start state have different outputs.
A DRMA can thus have multiple runs.
We should state that there is also another notion of determinism, 
similar to that in \cite{mohri2004weighted},
which is stricter and can be useful in some cases.
This notion requires at most one transition to be followed,
regardless of the output.
According to this strict definition,
the RMA of Fig. \ref{fig:example1}, e.g., is non-deterministic,
since both transitions from the start state can evaluate to \trueb.
By definition,
for this kind of determinism,
at most one run may exist for every stream.
We will use this notion of determinism in the next section.


We now give the definition for closure under union, concatenation, Kleene$+$, complement and determinization:
\begin{definition}[Closure of RMA]
\label{definition:closure_rmas}
We say that RMA are closed under: 
\begin{itemize}
	\item union if, for every RMA $A_{1}$ and $A_{2}$,
	there exists a RMA $A$ such that $\llbracket A \rrbracket(S) = \llbracket A_{1} \rrbracket(S)\cup \llbracket A_{2} \rrbracket(S)$ for every stream $S$, i.e., $M$ is a match of $A$ iff it is a match of $A_{1}$ or $A_{2}$.
	\item concatenation if, for every RMA $A_{1}$ and $A_{2}$,
	there exists a RMA $A$ such that \linebreak $\llbracket A \rrbracket(S) = \{M: M = M_{1} \cdot M_{2}, M_{1} \in \llbracket A_{1} \rrbracket(S), M_{2} \in \llbracket A_{2} \rrbracket(S) \}$ for every stream $S$, i.e., $M$ is a match of $A$ iff $M_{1}$ is a match of $A_{1}$, $M_{2}$ is a match of $A_{2}$ and $M$ is the concatenation of $M_{1}$ and $M_{2}$ (i.e., $M = M_{1} \cup M_{2}$ and $min(M_{2}) > max(M_{1})$).
	\item Kleene$+$ if, for every RMA $A$,
	there exists a RMA $A^{+}$ such that \linebreak $\llbracket A^{+} \rrbracket(S) = \{M: M = M_{1} \cdot M_{2} \cdots M_{n}, M_{i} \in \llbracket A \rrbracket(S) \}$ for every stream $S$, i.e., $M$ is a match of $A^{+}$ iff each $M_{i}$ is a match of $A$ and $M$ is the concatenation of all $M_{i}$.
	\item complement if, for every RMA $A$,
	there exists a RMA $A^{c}$ such that \linebreak $M \in \llbracket A \rrbracket(S) \Leftrightarrow M \notin \llbracket A^{c} \rrbracket(S)$.
	\item determinization if, for every RMA $A$,
	there exists a DRMA $A^{D}$ such that, \linebreak
	$M \in \llbracket A \rrbracket(S_{i}) \Leftrightarrow M \in \llbracket A^{D} \rrbracket(S_{i})$.
\end{itemize}
\end{definition}


For the closure properties of RMA, we have:
\begin{theorem}
\label{theorem:nrma_closure}
RMA are closed under concatenation, union, Kleene$+$ and determinization, but not under complement.
\end{theorem}
\begin{proof}[Proof sketch]
The proof for concatenation, union and Kleene$+$ follows from the proof of Theorem \ref{theorem:cepl2nrmae}.
The proof about complement is is essentially the same as that for register automata \cite{kaminski_finite-memory_1994}.
The proof for determinization is presented in the Appendix.
It is constructive and the determinization algorithm is based on the power--set construction of the states of the non--deterministic RMA and is similar to the algorithm for symbolic automata, 
but also takes into account the output of each transition.
It does not add or remove any registers.
It works in a manner very similar to the determinization algorithm for symbolic automata and MA \cite{veanes2010rex,DBLP:journals/corr/abs-1709-05369}.
It initially constructs the power set of the states of the URMA. 
The members of this power set will be the states of the DRMA.
It then tries to make each such new state, say $q^{d}$, deterministic,
by creating transitions with mutually exclusive formulas when they have the same output.
The construction of these mutually exclusive formulas is done
by gathering the formulas of all the transitions that have as their source a member of $q^{d}$.
Out of these formulas,
the set of min-terms is created,
i.e.,
the mutually exclusive conjuncts constructed from the initial formulas,
where each conjunct is a formula in its original or its negated form.
A transition is then created for each combination of a min-term with an output,
with $q^{d}$ being the source.
Then, only one transition with the same output can apply,
since these min-terms are mutually exclusive.
\end{proof}

RMA can thus be constructed from the three basic operators (union, concatenation and Kleene$+$) in a compositional manner,
providing substantial flexibility and expressive power for CEP applications.
However, as is the case for register automata \cite{kaminski_finite-memory_1994}, 
RMA are not closed under complement,
something which could pose difficulties for handling \emph{negation},
i.e.,
the ability to state that a sub-pattern should not happen for the whole pattern to be detected.
We reserve the treatment of negation for future work.

\section{Windowed Expressions and Output--agnostic Determinism}
\label{section:cepl_model:wdrmas}

As already mentioned,
the notion of determinism that we have used thus far allows for multiple runs.
However, there are cases where a deterministic automaton with a single run is needed.
Having a single run offers the advantage of an easier handling of automata that work in a streaming setting,
since no clones need to be created and maintained for the multiple runs.
On the other hand,
deterministic automata with a single run are more expensive to construct before the actual processing can begin and can have exponentially more states than non--deterministic automata.
A more important application of deterministic automata with a single run for our line of work is when we need to \emph{forecast} the occurrence of complex events,
i.e.,
when we need to probabilistically infer when a pattern is expected to be detected
(see \cite{Alevizos2017} for an example of event forecasting,
using classical automata).
In this case,
having a single run allows for a direct translation of an automaton to a Markov chain \cite{nuel_pattern_2008},
a critical step for making probabilistic inferences about the automaton's run-time behavior.
Capturing the behavior of automata with multiple runs through Markov chains could possibly be achieved,
although it could require techniques, 
like branching processes \cite{gallager2012discrete}, 
in order to capture the cloning of new runs and killing of expired runs. 
This is a research direction we would like to explore,
but in this paper we will try to investigate whether a transformation of a non--deterministic RMA to a deterministic RMA with a single run is possible.
We will show that this is indeed possible if we add \emph{windows} to CEPL expressions and ignore the output of the transitions.
Ignoring the output of transitions is a reasonable restriction for forecasting,
since we are only interested about when a pattern is detected and not about which specific input events constitute a match.

We first introduce the notion of output--agnostic determinism:
\begin{definition}[Output--agnostic determinism]
\label{definition:output_free_determinism}
Let $A$ be a RMA and $q$ a state of $A$. 
We say that $A$ is output--agnostic deterministic if for all transitions $(q,f_{1},rs_{1})\rightarrow(p_{1},R_{1},o_{1})$,
$(q,f_{2},rs_{2})\rightarrow(p_{2},R_{2},o_{2})$ 
(transitions from the same state $q$, regardless of the output)
$f_{1}$ and $f_{2}$ are mutually exclusive.
We say that a RMA $A$ is output--agnostic determinizable if there exists an output--agnostic DRMA $A^{D}$ such that, there exists an accepting run $\varrho$ of $A$ over a stream $S$ iff there exists an accepting run $\varrho^{D}$ of $A^{D}$ over $S$. 
\end{definition}
Thus, for this notion of determinism we treat RMA as recognizers and not as transducers.
Note also that, by definition,
an output--agnostic DRMA can have at most one run.

We can show that RMA are not in general determinizable under output--agnostic determinism:
\begin{theorem}
RMA are not determinizable under output--agnostic determinism.
\end{theorem}
\begin{proof}[Proof sketch]
Consider the RMA of Fig. \ref{fig:example1}.
For a stream of $m$ $T$ events, followed by one $H$ event with the same $id$,
this RMA detects $m$ matches,
regardless of the value of $m$,
since it is non-deterministic.
It can afford multiple runs and create clones of itself upon the appearance of every new $T$ event.
On the other hand,
an output--agnostic DRMA with $k$ registers is not able to handle such a stream in the case of $m>k$,
since it can have only a single run and can thus remember at most $k$ events.
\end{proof}

We can overcome this negative result,
by using windows in CEPL expressions and RMA.
In general,
CEP systems are not expected to remember every past event of a stream
and produce matches involving events that are very distant.
On the contrary,
it is usually the case that CEP patterns include an operator that limits the search space of input events,
through the notion of windowing.
This observation motivates the introduction of windowing in CEPL.

\begin{definition}[Windowed CEPL expression]
\label{definition:windowed_expressions}
A windowed CEPL expression is an expression of the form
$\phi := \psi\ \ceplwin\ w$,
where $\psi$ is a core--CEPL expression,
as in Definition \ref{definition:core_cepl_grammar},
and $w \in \mathbb{N}: w>0$. 
Given a match $M$, a stream $S$, and an
index $i \in \mathbb{N}$, we say that $M$ belongs to the evaluation
of $\phi:= \psi\ \ceplwin\ w$ over $S$ starting at $i$ and under the valuation
$v$, 
if $M \in \llbracket \psi \rrbracket(S,i,v)$ and $\mathit{max}(M)-\mathit{min}(M)<w$.
\end{definition}

The \ceplwin\ operator does not add any expressive power to CEPL.
We could use the index of an event in the stream as an event attribute
and then add \ceplfilter\ formulas in an expression which ensure that the difference between the index of the last event read and the first is no greater that $w$.
It is more convenient, however, to have an explicit operator for windowing.

It is easy to see that for windowed expressions we can construct an equivalent RMA.
In order to achieve our final goal,
which is to construct an output--agnostic DRMA,
we first show how we can construct a so-called unrolled RMA from a windowed expression:
\begin{lemma}
\label{lemma:unrolled_rma}
For every bounded and windowed core--CEPL expression there exists an equivalent unrolled RMA,
i.e., a RMA without any loops, except for a self-loop on the start state.
\end{lemma}
\begin{proof}[Algorithm sketch]
The full proof and the complete construction algorithm are presented in the Appendix.
Here, we provide only the general outline of the algorithm and an example.
Consider, e.g., the expression $\phi_{4} {:=} \phi_{1}\ \ceplwin\ w$,
a windowed version of Expression \eqref{expression:t_seq_h_filter_eq_id}.
Fig. \ref{fig:det:example_unrolled} shows the steps taken for constructing the equivalent unrolled RMA for this expression.
A simplified version of the determinization algorithm is shown in Algorithm \ref{algorithm:unrolling_simplified}.

\begin{algorithm}
\caption{Constructing unrolled RMA for windowed expression (simplified).}
\label{algorithm:unrolling_simplified}
\SetAlgoNoLine
\KwIn{Windowed core-CEPL expression $\phi := \psi\ \ceplwin\ w$}
\KwOut{Deterministic RMA $A_{\phi}$ equivalent to $\phi$}
$A_{\psi,\epsilon} \leftarrow \mathit{ConstructRMA}(\psi)$\; \label{algorithm:unrolling_simplified:line:rma} 
$A_{\psi} \leftarrow \mathit{EliminateEpsilon}(A_{\psi,\epsilon})$\; \label{algorithm:unrolling_simplified:line:eliminate} 
enumerate all walks of $A_{\psi}$ of length up to $w$; \label{algorithm:unrolling_simplified:line:unroll_start} \tcp{{\footnotesize Now unroll $A_{\psi}$ (Algorithm \ref{algorithm:unroll_cyclesk}).}}\
join walks through disjunction\;
collapse common prefixes\;
add loop-state with \trueb\ predicate on start state \label{algorithm:unrolling_simplified:line:unroll_end}\;
\end{algorithm}

The construction algorithm first produces a RMA as usual,
without taking the window operator into account
(see line \ref{algorithm:unrolling_simplified:line:rma} of Algorithm \ref{algorithm:unrolling_simplified}).
For our example,
the result would be the RMA of Fig. \ref{fig:example1}.
Then the algorithm eliminates any $\epsilon$-transitions (line \ref{algorithm:unrolling_simplified:line:eliminate}).
The next step is to use this RMA in order to create the equivalent unrolled RMA (URMA).
The rationale behind this step is that the window constraint essentially imposes an upper bound on the number of registers that would be required for a DRMA.
For our example,
if $w{=}3$,
then we know that we will need at least one register,
if a $T$ event is immediately followed by an $H$ event.
We will also need at most two registers,
if two consecutive $T$ events appear before an $H$ event.
The function of the URMA is to create the number of registers that will be needed,
through traversing the original RMA. 
Algorithm \ref{algorithm:unrolling_simplified} does this by enumerating all the walks of length up to $w$
on the RMA graph,
by unrolling any cycles.
Lines \ref{algorithm:unrolling_simplified:line:unroll_start} -- \ref{algorithm:unrolling_simplified:line:unroll_end} of Algorithm \ref{algorithm:unrolling_simplified} show this process in a simplified manner.
The URMA for our example is shown in Fig. \ref{fig:det:example_unrolled} for $w{=}2$ and $w{=}3$.
The actual algorithm does not perform an exhaustive enumeration, 
but incrementally creates the URMA,
by using the initial RMA as a generator of walks.
Every time we expand a walk,
we add a new transition, a new state and possibly a new register,
as clones of the original transition, state and register.
In our example,
we start by creating a clone of $q^{s}$ in Fig. \ref{fig:example1},
also named $q^{s}$ in Fig. \ref{fig:det:example_unrolled}.
From the start state of the initial RMA, 
we have two options.
Either loop in $q^{s}$ through the \trueb\ transition
or move to $q^{1}$ through the transition with the $f$ formula.
We can thus expand $q^{s}$ of the URMA with two new transitions:
from $q^{s}$ to $q_{t}$ and from $q^{s}$ to $q_{f}$ in Fig. \ref{fig:det:example_unrolled}.
We keep expanding the RMA this way until we reach final states and without exceeding $w$.
As a result, the final URMA has the form of a tree, 
whose walks and runs are of length up to $w$.
Finally, we add a \trueb\ self-loop on the start state
(not shown in Fig. \ref{fig:det:example_unrolled} to avoid clutter),
so that the RMA can work on streams.
This loop essentially allows the RMA to skip any number of events and start detecting a pattern at any stream index.
\end{proof}

\begin{figure}[t]
    \centering
    \begin{subfigure}[b]{0.80\textwidth}
        \includegraphics[width=\textwidth]{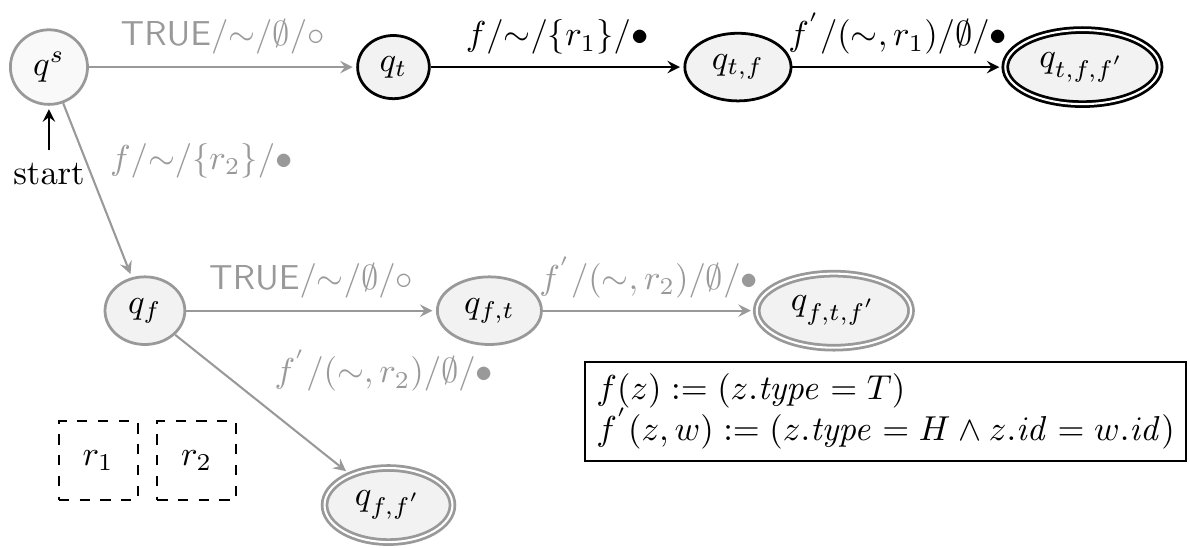}
        \caption{RMA after unrolling cycles, for $w=3$ (whole RMA, black and light gray states) and $w=2$ (top 3 states in black).}\label{fig:det:example_unrolled}
    \end{subfigure}
    ~ 
    \begin{subfigure}[b]{1.0\textwidth}
        \includegraphics[width=\textwidth]{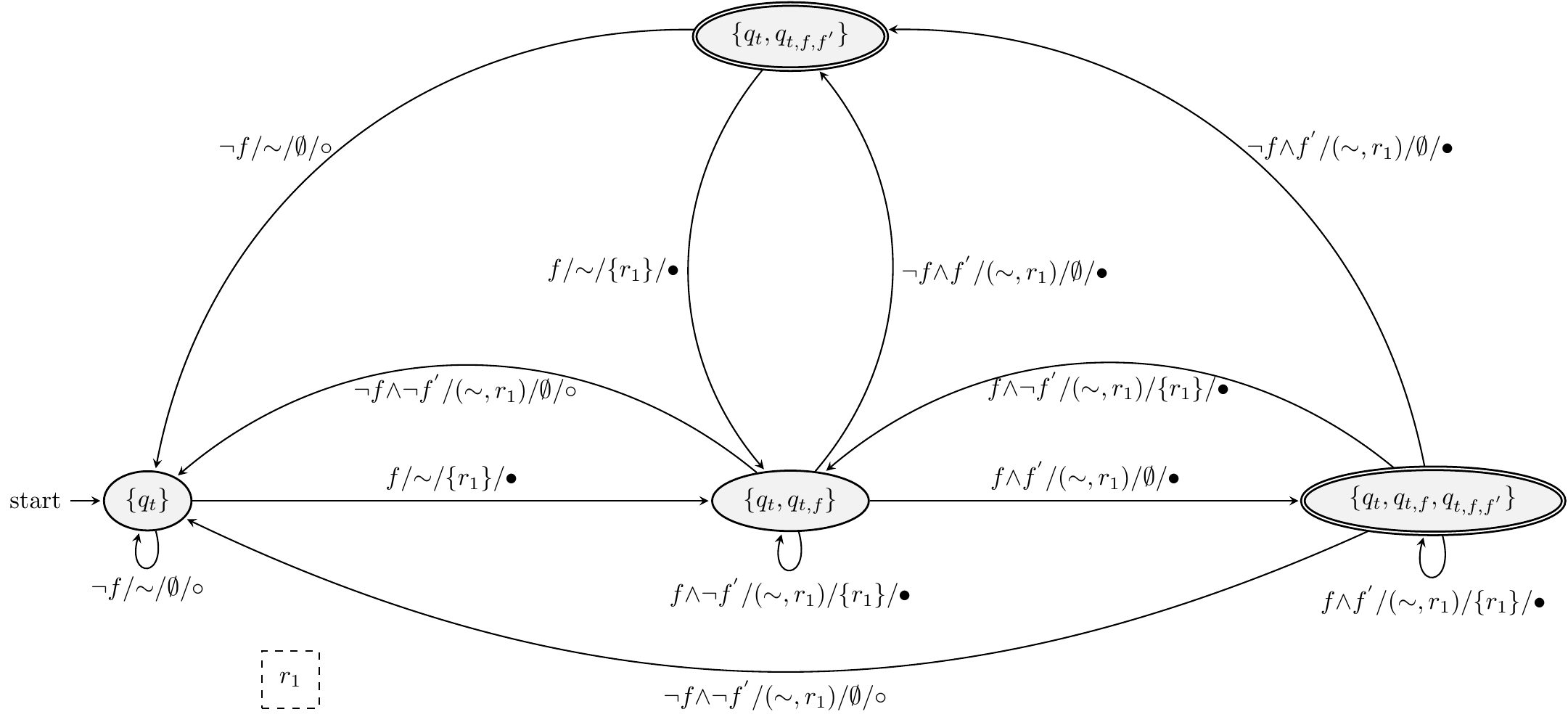}
        \caption{Output--agnostic DRMA, for $w=2$.}\label{fig:det:output_free}
    \end{subfigure}
    \caption{Constructing DRMA for $\phi_{1}\ \ceplwin\ w$.}\label{fig:det:example}
\end{figure}

A URMA then allows us to capture windowed expressions.
Note though that the algorithm we presented above,
due to the unrolling operation,
can result in a combinatorial explosion of the number of states of the DRMA,
especially for large values of $w$.
Its purpose here was mainly to establish Lemma \ref{lemma:unrolled_rma}.
In the future, 
we intend to explore more space-efficient techniques for constructing RMA equivalent to windowed expressions,
e.g.,
by incorporating directly the window constraint as a formula in the RMA.

Having a URMA makes it easy to subsequently construct an output--agnostic DRMA:
\begin{corollary}
\label{corollary:wcepl2drma}
Every URMA constructed from a bounded and windowed core--CEPL expression is output--agnostic determinizable.
\end{corollary}
\begin{proof}[Proof sketch]
In order to convert a URMA to an output--agnostic DRMA we modify the determinization algorithm so that the transition outputs are not taken into account.
Min--terms are constructed as in symbolic automata.
The proof about an accepting run of the URMA existing iff an accepting run of the output-agnostic DRMA exists is then the same as the proof for standard determinization.
The difference is that we cannot extend the proof to also state that the matches of the two RMA are the same,
since agnostic--output DRMA have a single--run and produce a single match,
whereas URMA produce multiple matches.
\end{proof}

As an example,
Fig. \ref{fig:det:output_free} shows the result of converting the URMA of Fig. \ref{fig:det:example_unrolled} to an output--agnostic DRMA
(only for $w{=}2$, due to space limitations).
We have simplified somewhat the formulas of each transition due to the presence of the \trueb\ predicates in some of them.
For example, the min-term $f {\wedge} \neg \trueb$ for the start state is unsatisfiable and can be ignored while $f {\wedge} \trueb$ may be simplified to $f$.
Note that,
as mentioned,
although the RMA of Figures \ref{fig:det:example_unrolled} and \ref{fig:det:output_free} are equivalent when viewed as recognizers,
they are not with respect to their matches.
For example, 
a stream of two $T$ events followed by an $H$ event
will be correctly recognized by both the URMA and the output--agnostic DRMA,
but the former will produce a match involving only the second $T$ event and the $H$ event,
whereas the latter will mark both $T$ events and the $H$ event.
However,
our final aim to construct a deterministic RMA with a single run that correctly detects when a pattern is completed has been achieved.

\section{Summary and Further Work}
\label{section:outro}

We presented an automaton model, RMA, that can act as a computational model for CEPL expressions with $n$-ary formulas,
which are quintessential for practical CEP applications. 
RMA thus extend the expressive power of MA and symbolic automata.
They also extend the expressive power of register automata,
through the use of formulas that are more complex than (in)equality predicates.
RMA have nice compositional properties, 
without imposing severe restrictions on the use of operators.
A significant fragment of core--CEPL expressions may be captured by RMA.
Moreover, we showed that outout--agnostic determinization is also possible,
if a window operator is used,
a very common feature in CEP.

As future work,
besides what has already been mentioned,
we need to investigate the class of CEPL expressions that can be captured by RMA,
since RMA are more expressive than bounded CEPL expressions.
We also intend to investigate how the extra operators (like negation) and the selection \emph{strategies} of CEPL may be incorporated.
We have presented here results about some basic closure properties. 
Other properties (e.g., decidability of emptiness, universality, equivalence, etc) remain to be determined,
although it is to be expected that RMA,
being more expressive than symbolic and register automata,
will have more undesirable properties in the general case,
unless restrictions are imposed,
like windowing, 
which helps in determinization. 
We also intend to do complexity analysis on the algorithms presented here and on the
behavior of RMA.
Last but not least,
it is important to investigate the relationship between RMA and other similar automaton models,
like automata in sets with atoms \cite{bojanczyk2011automata} and Quantified Event Automata \cite{barringer2012quantified}.

\bibliography{refs}

\appendix
\section{Appendix}
\label{section:appendix}

\begin{proof}[Proof of Theorem \ref{theorem:cepl2nrmae}]
\label{proof:theorem:cepl2nrmae:nary:if:short}
The proof is inductive and the algorithm compositional,
starting from the base case where $\phi := R\ \ceplas\ x\ \ceplfilter\ f(x)$.
Besides the base case,
there are four other cases to consider:
three for concatenation, union and Kleene$+$ and one more for filters with $n$-ary formulas.
The proofs and algorithms for the first four cases are very similar to the ones for Match Automata \cite{DBLP:journals/corr/abs-1709-05369}.
Here, we present the full proof for $n$-ary formulas and for one direction only, 
i.e.,
we prove the following:
$M \in \llbracket A_{\phi} \rrbracket(S_{i}) \Rightarrow \exists v: M \in \llbracket \phi \rrbracket(S,i,v)$ for $\phi := \phi^{'} \ceplfilter\ f(x_{1},\cdots,x_{n})$.
Algorithm \ref{algorithm:cepl2rma:nary} is the construction algorithm for this case.

\begin{algorithm}
\caption{CEPL to RMA for $n$-ary filter.}
\label{algorithm:cepl2rma:nary}
\SetAlgoNoLine
\KwIn{CEPL expression: $\phi = \phi^{'} \ceplfilter\ f(x_{1},\cdots,x_{n})$}
\KwOut{RMA $A$ equivalent to $\phi$ (and functions $\delta x$, $rx$)}
$(A_{\phi^{'}}, \delta x_{\phi^{'}}, rx_{\phi^{'}}) \leftarrow \mathit{ConstructRMA}(\phi^{'})$\;
$RG \leftarrow A_{\phi^{'}}.RG$\;
$rx \leftarrow rx_{\phi^{'}}$\;
\ForEach{$\delta \in A_{\phi^{'}}.\Delta$}{
	\If{$\exists x_{i} \in var(f): \delta x(\delta)=x_{i} \wedge \forall x_{j} \in var(f),\ x_{j}\ appears\ in\ every\ trail\ to\ \delta.\mathit{source} $}{ \label{line:cepl2rma:nary:isvalid}
		$\delta.f \leftarrow \delta.f \wedge f(xf_{1},\cdots,xf_{n})$\; \label{line:cepl2rma:nary:append_alpha}
		\tcc{{\footnotesize see Algorithm \ref{algorithm:cepl2rma:nary:create_new_rs} for $\mathit{CreateNewRs}$.}}\
		$(rs_{new}, RG_{new}, rx_{new}) \leftarrow \mathit{CreateNewRs}(A_{\phi^{'}}, \delta x_{\phi^{'}}, rx_{\phi^{'}}, \delta, f(x_{1}, \cdots, x_{n}))$\; \label{line:cepl2rma:nary:create_new_rs}
		$\delta.rs \leftarrow \delta.rs :: rs_{new}$\; \label{line:cepl2rma:nary:append_rs}
		$RG \leftarrow RG \cup RG_{new}$\;
		$rx \leftarrow rx \cup rx_{new}$\;
	}
}
$A \leftarrow (A_{\phi^{'}}.Q, A_{\phi^{'}}.q^{s}, A_{\phi^{'}}.Q^{f}, RG, A_{\phi^{'}}.\Delta)$\;
return $(A, \delta x_{\phi^{'}}, rx)$\; 
\end{algorithm}

\begin{algorithm}
\caption{$\mathit{CreateNewRs}$.}
\label{algorithm:cepl2rma:nary:create_new_rs}
\KwIn{A RMA $A$ (with functions $\delta x$ and $rx$), a transition $\delta$ and a formula $f(x_{1},\cdots,x_{2})$}
\KwOut{A new register selection $rs_{new}=(r_{1},\cdots,r_{n})$, a set of new registers $RG_{new}$ and a new function $rx_{new}$ for any new registers.
(also some transitions possibly modified).
}
$rs_{new} \leftarrow ()$\;
$RG_{new} \leftarrow \emptyset$\;
$rx_{new} \leftarrow \emptyset$\;
\ForEach{$x_{k} \in var(f)$}{
	\uIf{$\delta x(\delta)=x_{k}$}{
		$rs_{new} \leftarrow rs_{new} :: \sim$\;
	}
	\uElseIf{$x_{k} \in \mathit{range}(rx)$}{
		$rs_{new} \leftarrow rs_{new} :: rx^{-1}(x_{k})$\; \label{line:cepl2rma:nary:create_new_rs:r_exists}
	}
	\Else{
		$r_{new} \leftarrow \mathit{CreateNewRegister}()$\;
		$RG_{new} \leftarrow RG_{new} \cup \{r_{new}\}$\;
		$rs_{new} \leftarrow rs_{new} :: r_{new}$\;
		\ForEach{$\delta \in \delta x^{-1}(x_{k})$} { \label{line:cepl2rma:nary:create_new_rs:every_delta}
			$\delta.r \leftarrow r_{new}$\; \label{line:cepl2rma:nary:create_new_rs:new_r}
		}
		$rx_{new} \leftarrow rx_{new} \cup \{r_{new} \rightarrow x_{k}\}$\; \label{line:cepl2rma:nary:create_new_rs:new_assignment}
	}
}
return $(rs_{new},RG_{new},rx_{new})$\;
\end{algorithm}

First, some preliminary definitions are required.
During the construction of a RMA from a CEPL expression,
we keep and update two functions,
referring to the variables of the initial CEPL expression and how these are related to the transitions and registers of the RMA:
First, the partial function $\delta x: \Delta \rightharpoonup \mathbf{X}$,
mapping the transitions of the RMA to the variables of the CEPL expression;
Second, the total function $rx: RG \rightarrow \mathbf{X}$,
mapping the registers of the RMA to the variables of the CEPL expression.
With a slight abuse of notation, 
we will also sometimes use the notation $f^{-1}(y)$ to refer to all the domain elements of $f$ that map to $y$.

We also present some further properties that we will need to track
when constructing a RMA $A_{\phi}$ from a CEPL expression $\phi$.
At every inductive step,
we assume that the following properties hold for sub-expressions of $\phi$ and sub-automata of $A_{\phi}$,
except for the base case where it is directly proven that the properties hold.
At the end of every step,
we need to prove that these properties continue to hold for $\phi$ and $A_{\phi}$ as well.
The details of these proofs are omitted,
except for the case of $n$-ary formulas that we present here.

\begin{property}
\label{property:bound_variables_appearance}
For every walk $w$ induced by an accepting run and for every $x \in\ bound(\phi)$, $x$ appears exactly once in $w$.
Moreover, there also exists a trail $t$ contained in $w$ such that,
for every $x \in\ bound(\phi)$, $x$ appears exactly once in $t$.
\end{property}

\begin{property}
\label{property:registers_valuation}
Assume $M \in \llbracket A_{\phi} \rrbracket$, i.e., $\exists \varrho: match(\varrho)=M$, and $\exists v: M \in \llbracket \phi \rrbracket(S,i,v)$.
Let $\varrho$ be:
\begin{equation*}
\varrho = \cdots [i,q_{i},\gamma_{i}] \overset{\delta_{i}/o_{i}}{\rightarrow} [i+1,q_{i+1},\gamma_{i+1}] \cdots
\end{equation*}
and let $t_{i}=v_{S}(x_{b})$ be the tuple of the stream assigned to $x_{b} \in var(\phi)$ through valuation $v$.
Then, the following relationships hold between $v$ and $\varrho$:
\begin{itemize}
	\item For the transition $\delta_{i}$ that ``consumed'' $t_{i}$, it holds that
$\delta x(\delta_{i})=x_{b}$.
	\item Moreover, if $x_{b} \in bound(\phi)$ and is assigned to a register ($x_{b} \in \mathit{range}(rx)$ with $rx(r_{b})=x_{b}$), 
then, for each $\gamma_{j}$, it holds that
\begin{equation*}
\gamma_{j}(r_{b}) =
  \begin{cases}
  \sharp & j \leq i \\
  t_{i} &  j > i
  \end{cases}
\end{equation*}

\end{itemize}
\end{property}
In other words, an event from the stream associated with variable $x$ can only trigger transitions associated with this same variable.
Additionally, if $x$ is associated with a register,  
then the event will be written to that register once at position $i$.

\begin{property}
\label{property:unassigned_variables}
If $x \notin \mathit{range}(rx)$, then $\forall \delta \in \delta x^{-1}(x)$: $\delta.R=\emptyset$.
If $x \in \mathit{range}(rx)$, then $\forall \delta_{i},\delta_{j} \in \delta x^{-1}(x)$: $\delta_{i}.R=\delta_{j}.R \neq \emptyset$.
\end{property}
In other words, if a variable $x$ has not been assigned to a register,
all transitions associated with this variable do not write to any registers.
If a variable $x$ has been assigned to a register,
then all transitions associated with this variable write to the same register.

We first prove the fact that (detailed proof omitted):
$M \in \llbracket A_{\phi} \rrbracket(S_{i}) \Rightarrow M \in \llbracket A_{\phi^{'}} \rrbracket(S_{i})$
i.e., if $M$ is a match of $A_{\phi}$ then it should also be a match of $A_{\phi^{'}}$,
since $A_{\phi}$ is the same as $A_{\phi^{'}}$ but with more constraints on some of its transitions.
If a match can satisfy the constraints of $A_{\phi}$,
it should also satisfy the more relaxed constraints of $A_{\phi^{'}}$.
We can now conclude, by the induction hypothesis,
that:
$\exists v^{'}: M \in \llbracket \phi^{'} \rrbracket(S,i,v^{'})$.

We can try this valuation for $\phi$ as well.
We then need to prove that $M \in \llbracket \phi \rrbracket(S,i,v^{'})$.
By the definition of the CEPL semantics \cite{DBLP:journals/corr/abs-1709-05369}, 
we see that we need to prove two facts:
\begin{itemize}
	\item 
	$M \in \llbracket \phi^{'} \rrbracket(S,i,v^{'})$
	\item $v^{'}_{S} \vDash f$ or, equivalently, 
	\begin{equation}
	\label{to_prove_v_entails_alpha}
	f(v^{'}_{S}(x_{1}),\cdots,v^{'}_{S}(x_{n}))=\trueb
	\end{equation}		
\end{itemize}
The first one has already been proven.

We now need to prove the second one.
Note first, that,
the initial assumption $M \in \llbracket A_{\phi} \rrbracket(S_{i})$ means that there is an accepting run $\varrho$ such that $match(\varrho)=M$.
By Property \ref{property:bound_variables_appearance},
we can conclude that,
no matter what the accepting run $\varrho$ is,
it will have necessarily passed through a trail where every $x \in var(f)$ appears exactly once
(more precisely, where every $x \in bound(\phi)$ appears once,
and,
since $var(f) \subseteq bound(\phi)$, 
the same for every $x \in var(f)$).
This means that Algorithm \ref{algorithm:cepl2rma:nary} will have updated one transition on this trail, by ``appending'' $f$ to its original formula $f$ (line \ref{line:cepl2rma:nary:append_alpha}).
Moreover, since the run is accepting,
this transition will have applied.
More precisely,
if $\varrho$ is the accepting run and $\delta_{i}$ is this transition, the pair of successor configurations linked through  it would be:
\begin{equation*}
\varrho = \cdots [i,q_{i},\gamma_{i}] \overset{\delta_{i}/o_{i}}{\rightarrow} [i+1,q_{i+1},\gamma_{i+1}] \cdots
\end{equation*}
The formula of $\delta_{i}$ would then be $\delta_{i}.f = \delta^{'}_{i}.f \wedge f$,
where $\delta^{'}_{i}$ is the corresponding transition in $A_{\phi^{'}}$ 
(see again line \ref{line:cepl2rma:nary:append_alpha}).
Now, the fact that $\delta_{i}$ applied means that
\begin{equation*}
\delta_{i}.f(\gamma_{i}(\delta_{i}.rs))=\trueb \Rightarrow \delta^{'}_{i}.f(\gamma_{i}(rs_{old})) \wedge f(\gamma_{i}(rs_{new})) = \trueb
\end{equation*}
where $rs_{old}$ is the register selection of the transition in $A_{\phi^{'}}$
and $rs_{new}$ is the new register selection created for $f(x_{1},\cdots,x_{n})$ in
line \ref{line:cepl2rma:nary:create_new_rs}.
But this also implies that 
\begin{equation}
\label{fact:aplha_rs_true}
f(\gamma_{i}(rs_{new})) := f(\gamma_{i}(rs_{new}.1),\cdots,\gamma_{i}(rs_{new}.n))= \trueb
\end{equation}
Note the similarity between what we have established thus far in Eq. \eqref{fact:aplha_rs_true}
and what we need to prove in Eq. \eqref{to_prove_v_entails_alpha}.
We will now prove that $\gamma_{i}(rs_{new}.b)=v^{'}_{S}(x_{b}), \forall 1 \leq b \leq n$
and this will conclude our proof
(note that we will not deal with $\epsilon$ and \trueb\ transitions,
since they always apply and do not affect the registers.).

As we have shown,
if $\varrho$ is an accepting run of $A_{\phi}$,
then a run $\varrho^{'}$ of $A_{\phi^{'}}$ is induced which is also accepting:
\begin{equation*}
\varrho^{'} = \cdots [i,q_{i},\gamma^{'}_{i}] \overset{\delta^{'}_{i}/o_{i}}{\rightarrow} [i+1,q_{i+1},\gamma^{'}_{i+1}] \cdots
\end{equation*}
where the transitions $\delta_{i}$ of $A_{\phi}$ are the transitions $\delta^{'}_{i}$ of $A_{\phi^{'}}$,
possibly modified (in their formulas or writing registers) by Algorithm \ref{algorithm:cepl2rma:nary}
and 
\begin{equation}
\label{eq:gamma_entails_gamma_prime}
\gamma_{i} \vDash \gamma^{'}_{i}
\end{equation}
i.e., the contents of registers common to both RMA are the same.

Let $x_{b}$ be a variable of $f$ in Eq. \eqref{to_prove_v_entails_alpha}.
Since $\phi$ is bounded, $x_{b} \in bound(\phi^{'})$.
Let $t_{j}=v^{'}_{S}(x_{b})$ be the tuple assigned to $x_{b}$ by valuation $v^{'}$.
By the induction hypothesis and Property \ref{property:registers_valuation},
we know that, for $\varrho^{'}$ and $v^{'}$, the following hold:
\begin{itemize}
	\item For the transition $\delta^{'}_{j}$ that consumed $t_{j}$, 
	\begin{equation}
	\label{eq:deltaj_xj}
	\delta x^{'}(\delta^{'}_{j})=x_{b}
	\end{equation}
	\item If $x_{b} \in \mathit{range}(rx^{'})$ with $rx^{'}(r^{'}_{b})=x_{b}$:
	\begin{equation}
	\label{eq:gammak_xj}
	\gamma^{'}_{k}(r^{'}_{b}) =
  	\begin{cases}
  	\sharp & k \leq j \\
  	 t_{b} &  k > j
  	\end{cases}
	\end{equation}
\end{itemize}

As we have said,
the transition $\delta_{i}$ of $A_{\phi}$ is a transition that has been modified by ``appending'' the
formula $f$ to the formula $\delta^{'}_{i}$ of $A_{\phi^{'}}$.
However, note that Algorithm \ref{algorithm:cepl2rma:nary}
can do this ``appending'' only if one variable of $f$ (say $x_{m}$) is associated with $\delta^{'}_{i}$
and all other variables of $f$ appear in every trail to $\delta^{'}_{i}$ 
(more precisely, to its source state).
Let $w^{'}$ be the walk on $A_{\phi^{'}}$ induced by $\varrho^{'}$:
$w^{'} {=} <\cdots, \delta^{'}_{i}, \cdots>$.
We will now prove that no variables of $f$ can appear after $\delta^{'}_{i}$ in $w^{'}$.
Assume that one variable of $f$ does indeed appear after $\delta^{'}_{i}$.
Now, if we take the sub-walk of $w^{'}$ that ends at $\delta^{'}_{i}$:
$w^{'}_{i} {=} <\cdots, \delta^{'}_{i}>$,
we know (proof omitted)
that $w^{'}_{i}$ contains a trail to $\delta^{'}_{i}.source$.
But this trail will necessarily contain all variables $x \in var(f) - \{x_{m}\}$.
Therefore, if such a variable appears after $\delta^{'}_{i}.source$ as well,
it will appear at least twice in $w^{'}$.
But, since $x \in bound(\phi^{'})$ and $w^{'}$ is a walk induced by the accepting run $\varrho^{'}$,
by Property \ref{property:bound_variables_appearance},
this is a contradiction.
With respect to $x_{m}$,
by the same property,
we know that $x_{m}$ appears in $\delta^{'}_{i}$,
thus it cannot appear later.
Note that this is also true for $w$ and $\delta_{i}$,
since the two RMA are structurally the same and they have the same $\delta x$ functions.

Going back to $x_{b}$, we can refine Eq. \eqref{eq:deltaj_xj} and \eqref{eq:gammak_xj} to:
\begin{itemize}
	\item Either $x_{b}$ appears at $\delta^{'}_{i}$ ($j=i$), i.e., $\delta^{'}_{j}=\delta^{'}_{i}$.
	Thus 
	\begin{equation}
	\label{eq:delta_prime_xb}
	\delta x^{'}(\delta^{'}_{j})=\delta x^{'}(\delta^{'}_{i})=x_{b}
	\end{equation}		
	and if $x_{b} \in \mathit{range}(rx^{'})$ with $rx^{'}(r^{'}_{b})=x_{b}$
	\begin{equation}
	\gamma^{'}_{j}(r^{'}_{b}) = \gamma^{'}_{i}(r^{'}_{b}) = \sharp
	\end{equation}
	\item or $x_{b}$ appears before $\delta^{'}_{i}$ ($j<i$).
	Thus, if $x_{b} \in \mathit{range}(rx^{'})$ with $rx^{'}(r^{'}_{b})=x_{b}$:
	\begin{equation}
	\label{eq:gamma_rb_tj}
	\gamma^{'}_{i}(r^{'}_{b}) = t_{j} 
	\end{equation}
\end{itemize}

We can now check the different cases for the registers in Eq. \eqref{fact:aplha_rs_true}.
\begin{itemize}
	\item $rs_{new}.b = \sim$. 
	By definition, this means that $\gamma_{i}(rs_{new}.b)=\gamma_{i}(\sim)=t_{i}$.
	Note that $t_{i}$ (more precisely, its index $i$) belongs to the match $M$,
	i.e., $i \in match(\varrho)=M$ 
	and $i \in match(\varrho^{'})$ as well.
	This means that $i$ is the image of some variable $x_{bi}$ in the valuation $v^{'}$, i.e.,
	$v^{'}_{S}(x_{bi})=t_{i}$.
	By Property \ref{property:registers_valuation},
	this means that $\delta x^{'}(\delta^{'}_{i})=x_{bi}$.
	Additionally, Algorithm \ref{algorithm:cepl2rma:nary:create_new_rs} will return $\sim$
	for $rs_{new}.b$ only if $\delta x^{'}(\delta^{'}_{i})=x_{b}$.
	Therefore, $x_{bi}=x_{b}$ and $v^{'}_{S}(x_{b})=v^{'}_{S}(x_{bi})=t_{i}$.
	As a result, $\gamma_{i}(rs_{new}.b)=v^{'}_{S}(x_{b})=t_{i}$. 
	\item $rs_{new}.b \in A_{\phi^{'}}.RG$, i.e., this register is common to both RMA.
	By the construction Algorithm \ref{algorithm:cepl2rma:nary:create_new_rs},
	we know that $rx(rs_{new}.b)=rx^{'}(rs_{new}.b)=x_{b}$ and that $x_{b}$ appears before $\delta^{'}_{i}$. 
	Now, let $t_{j}=v^{'}_{S}(x_{b})$ be the tuple assigned to $x_{b}$ by valuation $v^{'}$.
	By Eq. \eqref{eq:gamma_rb_tj}, we know that $\gamma^{'}_{i}(rs_{new}.b)=t_{j}$.
	By Eq. \eqref{eq:gamma_entails_gamma_prime}, we also know that $\gamma_{i}(rs_{new}.b)=\gamma^{'}_{i}(rs_{new}.b)$.
	Therefore, $\gamma_{i}(rs_{new}.b)=v^{'}_{S}(x_{b})=t_{j}$.
	\item $rs_{new}.b \notin A_{\phi^{'}}.RG$, i.e., this is a new register.
	By the construction Algorithm \ref{algorithm:cepl2rma:nary:create_new_rs},
	we know that $rx(rs_{new}.b)=x_{b}$, $x_{b} \notin \mathit{range}(rx^{'})$ and that $x_{b}$ appears before $\delta^{'}_{i}$ / $\delta_{i}$. 
	Now, let $t_{j}=v^{'}_{S}(x_{b})$ be the tuple assigned to $x_{b}$ by valuation $v^{'}$
	and $\delta^{'}_{j}$ the transition (before $\delta^{'}_{i}$) that consumed $t_{j}$.
	By Eq. \eqref{eq:delta_prime_xb}, we know that $\delta x^{'}(\delta^{'}_{j})=x_{b}$.
	But Algorithm \ref{algorithm:cepl2rma:nary:create_new_rs} will have updated $\delta^{'}_{j}$
	so that $\delta_{j}.R=\{rs_{new}.b\}$.
	This means that $\delta_{j}$ will write $t_{j}$ to $rs_{new}.b$,
	thus 
	\begin{equation*}
	\gamma_{k}(rs_{new}.b) =
  	\begin{cases}
  	\sharp & k \leq j \\
  	 t_{j} &  k > j
  	\end{cases}
	\end{equation*}
	(reminder: $x_{b}$ appears only once in $w$ which means that $rs_{new}.b$ will be written only once at $j$).
	Since $i>j$, $\gamma_{i}(rs_{new}.b)=t_{j}$, which implies that $\gamma_{i}(rs_{new}.b)=v^{'}_{S}(x_{b})$.
\end{itemize}

With respect to Property \ref{property:registers_valuation},
note that it also holds for $\varrho$ of $A_{\phi}$ 
and $v^{'}$,
as a valuation for $\phi$.
By the induction hypothesis,
it holds for $\varrho^{'}$ of $A_{\phi^{'}}$ 
and $v^{'}$,
as a valuation for $\phi^{'}$.
But the corresponding transitions in $\varrho$
are the same as those of $\varrho^{'}$,
as far as their associated variables are concerned
($\delta x$ remains the same).
Additionally, $v^{'}$,
as we have just proved
is a valuation for $\phi$ as well.
Therefore, the first part of the property holds.
The second part has been proven just above.
We just note that this part holds for the case where $rs_{new}.b \in A_{\phi^{'}}.RG$ as well,
by the induction hypothesis.

\end{proof}

\begin{proof}[Proof of Theorem \ref{theorem:nrma_closure} for determinization]
\label{proof:theorem:wcepl2drma:short}
The process for constructing a deterministic RMA (DRMA) from a CEPL expression is shown in Algorithm \ref{algorithm:determinization_multi_run}.
It first constructs a non-deterministic RMA (NRMA) and then uses the power set of this NRMA's states to construct the DRMA.
For each state $q^{D}$ of the DRMA,
it gathers all the formulas from the outgoing transitions of the states of the NRMA $q^{N}$ ($q^{N} \in q^{D}$),
it creates the (mutually exclusive) marked min--terms from these formulas
and then creates transitions, based on these min--terms.
A min--term is called marked when the output, $\bullet$ or $\circ$,
is also taken into account.
For each original min--term,
we have two marked min--terms,
one where the output is $\bullet$ and one where the output is $\circ$.
Please, note that this is the first time that we use the ability of a transition to write to more than one registers.
So, from now on,
$\delta.R$ will be a set that is not necessarily a singleton.
This allows us to retain the same set of registers, i.e.,
the set of registers $RG$ will be the same for the NRMA and the DRMA.
A new transition created for the DRMA may write to multiple registers,
if it ``encodes'' multiple transitions of the NRMA,
which may write to different registers.
It is also obvious that the resulting RMA is deterministic,
since the various min--terms out of every state are mutually exclusive for the same output.

First, we will prove the following proposition:
There exists a run $\varrho^{N} \in Run_{k}(A^{N},S(i))$ that $A^{N}$ can follow by reading the first $k$ tuples from the sub-stream $S(i)$,
iff there exists a run $\varrho^{D} \in Run_{k}(A^{D},S(i))$ 
that $A^{D}$ can follow by reading the same first $k$ tuples,
such that, if
\begin{equation*}
\varrho^{N} = [i,q_{i}^{N},\gamma_{i}^{N}] \overset{\delta_{i}^{N}/o_{i}^{N}}{\rightarrow} [i+1,q_{i+1}^{N},\gamma_{i+1}^{N}] \cdots [i+k,q_{i+k}^{N},\gamma_{i+k}^{N}]
\end{equation*}
and
\begin{equation*}
\varrho^{D} = [i,q_{i}^{D},\gamma_{i}^{D}] \overset{\delta_{i}^{D}/o_{i}^{D}}{\rightarrow} [i+1,q_{i+1}^{D},\gamma_{i+1}^{D}] \cdots [i+k,q_{i+k}^{D},\gamma_{i+k}^{D}]
\end{equation*}
are the runs of $A^{N}$ and $A^{D}$ respectively, then,
\begin{itemize}
	\item $q_{j}^{N} \in q_{j}^{D}\ \forall j: i \leq j \leq i+k$
	\item if $r \in A_{D}.RG$ appears in $\varrho^{N}$, then it appears in $\varrho^{D}$
	\item $\gamma_{j}^{N}(r) = \gamma_{j}^{D}(r)$  for every $r$ that appears in $\varrho^{N}$ (and $\varrho^{D})$
\end{itemize}

We say that a register $r$ appears in a run at position $m$ if $r \in \delta_{m}.R$.

We will prove only direction (the other is similar).
Assume there exists a run $\varrho^{N}$.
We will prove that there exists a run $\varrho^{D}$ by induction on the length $k$ of the run.

\textbf{Base case: $k=0$.}
Then $\varrho^{N}=[i,q_{i}^{N},\gamma_{i}^{N}]=[i,q^{s,N},\gamma^{s,N}]$.
The run $\varrho^{D}=[i,q^{s,D},\gamma^{s,D}]$ is indeed a run of the DRMA
that satisfies the proposition,
since $q^{s,N} \in q^{s,D} = \{q^{s,N}\}$ 
(by the construction algorithm, line \ref{line:determinization_multi_run:start_state}),
all registers are empty
and no registers appear in the runs. 

\textbf{Case $k>0$.}
Assume the proposition holds for $k$.
We will prove it holds for $k+1$ as well.
Let $\varrho_{k}^{N}$ be a run of $A^{N}$ after the first $k$ tuples and
\begin{equation*}
\varrho_{k+1}^{N} = \cdots [i+k,q_{i+k}^{N},\gamma_{i+k}^{N}] \begin{cases}
\overset{\delta_{i+k}^{N,1}/\bullet}{\rightarrow} [i+k+1,q_{i+k+1}^{N,1},\gamma_{i+k+1}^{N,1}] \\
\cdots \\
\overset{\delta_{i+k}^{N,m}/\bullet}{\rightarrow} [i+k+1,q_{i+k+1}^{N,m},\gamma_{i+k+1}^{N,m}] \\
\overset{\delta_{i+k}^{N,m+1}/\circ}{\rightarrow} [i+k+1,q_{i+k+1}^{N,m+1},\gamma_{i+k+1}^{N,m+1}] \\
\cdots \\
\overset{\delta_{i+k}^{N,n}/\circ}{\rightarrow} [i+k+1,q_{i+k+1}^{N,n},\gamma_{i+k+1}^{N,n}] \\
\end{cases}
\end{equation*}
be the possible runs of the NRMA after reading $k+1$ tuples and expanding $\varrho_{k}^{N}$.
Then, we need to find a run of the DRMA like:
\begin{equation*}
\varrho_{k+1}^{D} = \cdots [i+k,q_{i+k}^{D},\gamma_{i+k}^{D}] \overset{\delta_{i+k}^{D}/o_{i+k}^{D}}{\rightarrow} [i+k+1,q_{i+k+1}^{D},\gamma_{i+k+1}^{D}]
\end{equation*}
Consider first the transitions whose output is $\bullet$.
By the induction hypothesis,
we know that $q_{i+k}^{N} \in q_{i+k}^{D}$.
By the construction Algorithm \ref{algorithm:determinization_multi_run},
we then know that,
if $f_{k+1}^{N,j}=\delta_{i+k}^{N,j}.f$ is the formula of a transition that takes the non-deterministic run to $q_{i+k+1}^{N,j}$ and outputs a $\bullet$,
then there exists a transition $\delta_{i+k}^{D}$ in the DRMA from $q_{i+k}^{D}$ whose formula will be a min--term,
containing all the $f_{k+1}^{N,j}$ in their positive form
and all other possible formulas in their negated form.
Moreover, the target of that transition, $q_{i+k+1}^{D}$, contains all $q_{i+k+1}^{N,j}$.
More formally,  $q_{i+k+1}^{D} = \bigcup\limits_{j=1}^{m}{q_{i+k+1}^{N,j}}$.
We also need to prove that $\delta_{i+k}^{D}$ applies as well.
As we said,
the formula of this transition is a conjunct,
where all $f_{k+1}^{N,j}$ appear in their positive form and all other formulas of in their negated form.
But note that the formulas in negated form are those that did not apply in $\varrho_{k+1}^{N}$ when reading the $(k+1)^{th}$ tuple.
Additionally,
the arguments passed to each of the formulas of the min--term are the same (registers) as those passed to them in the non-deterministic run
(by the construction algorithm, line \ref{line:determinization_multi_run:register_selection}). 
To make this point clearer,
consider the following simple example of a min--term
(where we have simplified notation and use registers directly as arguments):
\begin{equation*}
f = f_{1}(r_{1,1},\cdots,r_{1,k}) \wedge \neg f_{2}(r_{2,1},\cdots,r_{2,l}) \wedge f_{3}(r_{3,1},\cdots,r_{3,m})
\end{equation*}
This means that $f_{1}(r_{1,1},\cdots,r_{1,k})$,
with the exact same registers as arguments,
will be the formula of a transition of the NRMA that applied.
Similarly for $f_{3}$.
With respect to $f_{2}$,
it will be the formula of a transition that did not apply.
If we can show that the contents of those registers are the same in the runs of the NRMA and DRMA when reading the last tuple,
then this will mean that $\delta_{i+k}^{D}$ indeed applies.
But this is the case by the induction hypothesis
($\gamma_{i+k}^{N}(r) = \gamma_{i+k}^{D}(r)$),
since all these registers appear in the run $\varrho_{k}^{N}$ up to $q_{i+k}^{N}$.
The second part of the proposition also holds,
since,
by the construction,
$\delta_{i+k}^{D}$ will write to all the registers that the various $\delta_{i+k}^{N,j}$ write
(see line \ref{line:determinization_multi_run:register_writing} in Algorithm).
The third part also holds,
since $\delta_{i+k}^{D}$ will write the same tuple to the same registers as the various $\delta_{i+k}^{N,j}$.
By the same reasoning, we can prove the proposition for transitions whose output is $\circ$.

Since the above proposition holds for accepting runs as well,
we can conclude that there exists an accepting run of $A_{N}$ iff there exists an accepting run of $A_{D}$.
According to the above proposition,
the union of the last states over all $\varrho^{N}$ is equal to the last state of $\varrho^{D}$.
Thus, if $\varrho^{N}$ reaches a final state,
then the last state of $\varrho^{D}$ will contain this final state and hence be itself a final state.
Conversely, if $\varrho^{D}$ reaches a final state of $A_{D}$,
it means that this state contains a final state of $A_{N}$.
Then, there must exist a $\varrho^{N}$ that reached this final state.

What we have proven thus far is that $\varrho^{N}$ is accepting iff $\varrho^{D}$ is accepting.
Therefore, the two RMA are indeed equivalent if viewed as recognizers / acceptors.
Note, however, that the two runs, 
$\varrho_{k+1}^{N}$ and $\varrho_{k+1}^{D}$,
mark the stream at the same positions.
Therefore, if they are accepting runs, 
they produce the same matches,
i.e.,
if $M \in \llbracket A^{N} \rrbracket(S_{i})$,
then $M \in \llbracket A^{D} \rrbracket(S_{i})$.
\end{proof}

\begin{algorithm}
\caption{Determinization.}
\label{algorithm:determinization_multi_run}
\KwIn{Bounded core--CEPL expression $\phi$}
\KwOut{Deterministic RMA $A^{D}$ equivalent to $\phi$}
$A^{N} \leftarrow \mathit{ConstructRMA}(\phi)$\; 
$Q^{D} \leftarrow \mathit{ConstructPowerSet}(A^{N}.Q)$\;
$\Delta^{D} \leftarrow \emptyset$;	$Q^{f,D} \leftarrow \emptyset$\;
\ForEach{$q^{D} \in Q^{D}$}{
	\If{$q^{D} \cap A^{N}.Q^{f} \neq \emptyset$}{
		$Q^{f,D} \leftarrow Q^{f,D} \cup q^{D}$\;
	}
	$\mathit{Formulas} \leftarrow ()$;	$rs^{D} \leftarrow ()$\;
	\ForEach{$q^{N} \in q^{D}$}{
		\ForEach{$\delta^{N} \in A^{N}.\Delta: \delta^{N}.\mathit{source} = q^{N}$ }{
			$\mathit{Formulas} \leftarrow \mathit{Formulas} :: \delta^{N}.f$\;
			$rs^{D} \leftarrow rs^{D} :: \delta^{N}.rs$\; \label{line:determinization_multi_run:register_selection}
		}
	}
	$\mathit{MarkedMinTerms} \leftarrow \mathit{ConstructMarkedMinTerms}(\mathit{Formulas},\{ \bullet, \circ \})$\; 
	\ForEach{$\mathit{mmt} \in \mathit{MarkedMinTerms}$}{
		$p^{D} \leftarrow \emptyset$;	$R^{D} \leftarrow \emptyset$\;
		\ForEach{$q^{N} \in q^{D}$}{
			\ForEach{$\delta^{N} \in A^{N}.\Delta: \delta^{N}.\mathit{source} = q^{N}$ }{
				\If{$\mathit{mmt} \vDash \delta^{N}.f \wedge \mathit{mmt}.o=\delta^{N}.o$}{
					$p^{D} \leftarrow p^{D} \cup \{\delta^{N}.\mathit{target}\}$\;
					$R^{D} \leftarrow R^{D} \cup \{\delta^{N}.R\}$\; \label{line:determinization_multi_run:register_writing}
					$o^{D} \leftarrow \mathit{mmt}.o$\;
				}
			}
		}
		$\delta^{D} \leftarrow \mathit{CreateNewTransition}((q^{D},mt,rs^{D}) \rightarrow (p^{D},rw^{D},o^{D}))$\;
		$\Delta^{D} \leftarrow \Delta^{D} \cup \{\delta^{D}\}$\;
	}
}
$q^{s,D} \leftarrow \{A^{N}.q^{s}\}$\; \label{line:determinization_multi_run:start_state}
$A^{D} \leftarrow (Q^{D},q^{s,D},Q^{f,D},A^{N}.RG,\Delta^{D})$\;
return $A^{D}$\;
\end{algorithm}


\begin{proof}[Proof of Lemma \ref{lemma:unrolled_rma}]
Let $\phi := \psi\ \ceplwin\ w$.
Algorithm \ref{algorithm:wcepl2rma} shows how we can construct $A^{\phi}$
(we use superscripts to refer to expressions and reserve subscripts for referring to stream indexes in the proof).
The basic idea is that we first construct as usual the RMA $A^{\psi}$ for the sub-expression $\psi$
(and eliminate $\epsilon$-transitions).
We can then use $A^{\psi}$ to enumerate all the possible walks of $A^{\psi}$ of length up to $w$ and then join them in a single RMA through disjunction.
Essentially,
we need to remove cycles from every walk of $A^{\psi}$ by ``unrolling'' them as many times as necessary,
without the length of the walk exceeding $w$.
This ``unrolling'' operation is performed by the (recursive) Algorithm \ref{algorithm:unroll_cyclesk}.
Because of this ``unrolling'',
a state of $A^{\psi}$ may appear multiple times as a state in $A^{\phi}$.
We keep track of which states of $A^{\phi}$ correspond to states of $A^{\psi}$ through the function $\mathit{CopyOfQ}$ in the algorithm.
For example, if $q^{\psi}$ is a state of $A^{\psi}$, $q^{\phi}$ a state of $A^{\phi}$ and
$\mathit{CopyOfQ(q^{\phi}) = q^{\psi}}$,
this means that $q^{\phi}$ was created as a copy of $q^{\psi}$
(and multiple states of $A^{\phi}$ may be copies of the same state of $A^{\psi}$).
We do the same for the registers as well, through the function $\mathit{CopyOfR}$.
The algorithm avoids an explicit enumeration,
by gradually building the automaton as needed,
through an incremental expansion.
Of course, walks that do not end in a final state may be removed,
either after the construction or online,
whenever a non-final state cannot be expanded.

\begin{algorithm}
\caption{Constructing RMA for windowed CEPL expression.}
\label{algorithm:wcepl2rma}
\SetAlgoNoLine
\KwIn{Windowed core-CEPL expression $\phi := \psi\ \ceplwin\ w$}
\KwOut{RMA $A^{\phi}$ equivalent to $\phi$}
$(A^{\psi,\epsilon}, \delta x^{\psi,\epsilon}, rx^{\psi,\epsilon}) \leftarrow ConstructRMA(\psi)$\; 
$A^{\psi} \leftarrow \mathit{EliminateEpsilon}(A^{\psi,\epsilon})$\; 
$A^{\phi} \leftarrow Unroll(A^{\psi},w)$; \tcp{{\footnotesize see Algorithm \ref{algorithm:unroll_cyclesk}}}\ 
$\delta^{loop} \leftarrow \mathit{CreateNewTransition}((A^{\phi}.q^{s},\trueb,\sim) \rightarrow (A^{\phi}.q^{s},\emptyset,\circ))$\;
$A^{\phi}.\Delta \leftarrow A^{\phi}.\Delta \cup \{\delta^{loop}\}$\;
return $A^{\phi}$\;
\end{algorithm}


\begin{algorithm}
\caption{Unrolling cycles for windowed expressions, $k > 0$.}
\label{algorithm:unroll_cyclesk}
\KwIn{RMA $A$ and integer $k > 0$}
\KwOut{RMA $A_{k}$ with runs of length up to $k$}
$(A_{k-1},FrontStates,CopyOfQ,CopyOfR) \leftarrow Unroll(A,k-1)$\;
$NextFronStates \leftarrow \emptyset$\;
$Q_{k} \leftarrow A_{k-1}.Q$; $Q_{k}^{f} \leftarrow A_{k-1}.Q^{f}$\;
$RG_{k} \leftarrow A_{k-1}.RG$; $\Delta_{k} \leftarrow A_{k-1}.\Delta$\;
\ForEach{$q \in FrontStates$}{
	$q_{c} \leftarrow CopyOfQ(q)$\;
	\ForEach{$\delta \in A.\Delta: \delta.source = q_{c}$}{
		$q_{new} \leftarrow CreateNewState()$\;
		$Q_{k} \leftarrow Q_{k} \cup \{q_{new}\}$\;
		$CopyOfQ \leftarrow CopyOfQ \cup \{ q_{new} \rightarrow \delta.target \}$\;
		\If{$\delta.target \in A.Q^{f}$}{
			$Q_{k}^{f} \leftarrow Q_{k}^{f} \cup \{q_{new}\}$\;
		}	
		\uIf{$\delta.R = \emptyset$}{
			$R_{new} \leftarrow \emptyset$\;	
		}
		\Else{
			$r_{new} \leftarrow CreateNewRegister()$\;
			$RG_{k} \leftarrow RG_{k} \cup \{ r_{new} \}$\;
			$R_{new} \leftarrow \{ r_{new} \}$\;
			$CopyOfR \leftarrow CopyOfR \cup \{r_{new} \rightarrow \delta.r \}$; \label{line:unroll_cycles:rn_descendent}
			\tcp{{\footnotesize $\delta.r$ single element of $\delta.R$}}
		}		
		$f_{new} \leftarrow \delta.f$\;
		$o_{new} \leftarrow \delta.o$\;
		$rs_{new} \leftarrow ()$\;
		\ForEach{$r \in \delta.rs$}{
			$r_{latest} \leftarrow FindLastAppearance(r,q,A_{k-1})$\; \label{line:unroll_cycles:latest_appearance}
			$rs_{new} \leftarrow rs_{new} :: r_{latest}$\;
		}
		$\delta_{new} \leftarrow \mathit{CreateNewTransition}((q,f_{new},rs_{new}) \rightarrow (q_{new},R_{new},o_{new}))$\;
		$\Delta_{k} \leftarrow \Delta_{k} \cup \{ \delta_{new} \}$\;
		$NextFrontStates \leftarrow NextFrontStates \cup \{q_{new}\}$\;
	}
}
$A_{k} \leftarrow (Q_{k}, A_{k-1}.q^{s}, Q_{k}^{f}, RG_{k}, \Delta_{k})$\;
return $(A_{k},NextFrontStates,CopyOfQ,CopyOfR)$\;
\end{algorithm}

\begin{algorithm}
\caption{Unrolling cycles for windowed expressions, base case: $k=0$.}
\label{algorithm:unroll_cycles0}
\SetAlgoNoLine
\KwIn{RMA $A$}
\KwOut{RMA $A_{0}$ with runs of length 0}
$q \leftarrow CreateNewState()$\;
$CopyOfQ \leftarrow \{q \rightarrow A.q^{s}\}$\; \label{line:unroll_cycles:qs_descendent}
$CopyOfR \leftarrow \emptyset$\;
$FrontStates \leftarrow \{q\}$\;
$Q^{f} \leftarrow \emptyset$\;
\If{$A.q^{s} \in A.Q^{f}$}{
	$Q^{f} \leftarrow Q^{f} \cup \{q\}$\;
}
$A_{0} \leftarrow (\{q\},q,Q^{f},\emptyset,\emptyset)$\;
return $(A_{0},FrontStates,CopyOfQ,CopyOfR)$\;
\end{algorithm}

The lemma is a direct consequence of the construction algorithm.
First, note that,
by the construction algorithm,
there is a one-to-one mapping (bijective function) between the walks/runs of $A^{\phi}$
and the walks/runs of $A^{\psi}$ of length up to $w$.
We can show that if $\varrho^{\psi}$ is a run of $A^{\psi}$ of length up to $w$,
then the corresponding run $\varrho^{\phi}$ of $A^{\phi}$ is indeed a run,
with $match(\varrho^{\psi})=match(\varrho^{\phi})=M$,
where,
by definition, 
since the runs have no $\epsilon$-transitions and are at most of length $w$,
$max(M)-min(M)<w$.

We first prove the following proposition:
There exists a run $\varrho^{\psi}$ of $A^{\psi}$ of length up to $w$ iff
there exists a run $\varrho^{\phi}$ of $A^{\phi}$ such that:
\begin{itemize}
	\item $\mathit{CopyOfQ}(q_{j}^{\phi}) = q_{j}^{\psi}$
	\item $\gamma_{j}^{\psi}(r^{\psi})=\gamma_{j}^{\phi}(r^{\phi})$, 
	if
	$\mathit{CopyOfR}(r^{\phi})=r^{\psi}$
	and
	$r^{\phi}$ appears last among the registers that are copies of $r^{\psi}$ in $\varrho^{\phi}$.
\end{itemize}

We say that a register $r$ appears in a run at position $m$ if $r \in \delta_{m}.R$,
i.e.,
if the $m^{th}$ transition writes to $r$.
The notion of a register's (last) appearance also applies for walks of $A^{\phi}$,
since $A^{\phi}$ is a directed acyclic graph,
as can be seen by Algorithms \ref{algorithm:unroll_cycles0} and \ref{algorithm:unroll_cyclesk}
(they always expand ``forward'' the RMA, 
without creating any cycles and without converging any paths).

The proof is by induction on the length of the runs $k$, with $k \leq w$.
We prove only one direction (assume a run $\varrho^{\psi}$ exists).
The other is similar.

\textbf{Base case: $k=0$.}
For both RMA, only the start state and the initial configuration with all registers empty is possible.
Thus, $\gamma_{i}^{\psi}=\gamma_{i}^{\phi}=\sharp$ for all registers. 
By Algorithm \ref{algorithm:unroll_cycles0} (line \ref{line:unroll_cycles:qs_descendent}),
we know that $\mathit{CopyOf}(q^{s,\phi}) = q^{s,\psi}$.

\textbf{Case for $0 < k+1 \leq w$.}
Let
\begin{equation*}
\varrho_{k+1}^{\psi} = \cdots [i+k,q_{i+k}^{\psi},\gamma_{i+k}^{\psi}] \overset{\delta_{i+k}^{\psi}/o_{i+k}^{\psi}}{\rightarrow} [i+k+1,q_{i+k+1}^{\psi},\gamma_{i+k+1}^{\psi}]
\end{equation*}
and
\begin{equation*}
\varrho_{k+1}^{\phi} = \cdots [i+k,q_{i+k}^{\phi},\gamma_{i+k}^{\phi}] \overset{\delta_{i+k}^{\phi}/o_{i+k}^{\phi}}{\rightarrow} [i+k+1,q_{i+k+1}^{\phi},\gamma_{i+k+1}^{\phi}]
\end{equation*}
be the runs of $A^{\psi}$ and $A^{\phi}$ respectively of length $k+1$ over the same $k+1$ tuples.
We know that $\varrho_{k+1}^{\psi}$ is an actual run and we need to construct $\varrho_{k+1}^{\phi}$, 
knowing, by the induction hypothesis,
that it is an actual run up to $q_{i+k}^{\phi}$.
Now, by the construction algorithm,
we can see that if $\delta_{i+k}^{\psi}$ is a transition of $A^{\psi}$ from $q_{i+k}^{\psi}$ to $q_{i+k+1}^{\psi}$,
there exists a transition $\delta_{i+k}^{\phi}$ with the same formula and output
from $q_{i+k}^{\phi}$ to a $q_{i+k+1}^{\phi}$ such that $CopyOfQ(q_{i+k+1}^{\phi})=q_{i+k+1}^{\psi}$.
Moreover, if $\delta_{i+k}^{\psi}$ applies,
so does $\delta_{i+k}^{\phi}$,
because the registers in the register selection of $\delta_{i+k}^{\phi}$
are copies of the corresponding registers in $\delta_{i+k}^{\psi}.rs$.
By the induction hypothesis,
we know that the contents of the registers in $\delta_{i+k}^{\psi}.rs$ will be equal to the contents of their corresponding registers in $\varrho^{\phi}$ that appear last.
But these are exactly the registers in $\delta_{i+k}^{\phi}.rs$
(see line \ref{line:unroll_cycles:latest_appearance} in Algorithm \ref{algorithm:unroll_cyclesk}).
We can also see that the part of the proposition concerning the $\gamma$ functions also holds.
If $\delta_{i+k}^{\psi}.R = \{ r^{\psi} \}$ and $\delta_{i+k}^{\phi}.R = \{ r^{\phi} \}$,
then we know,
by the construction algorithm
(line \ref{line:unroll_cycles:rn_descendent}),
that $CopyOfR(r^{\phi}) = r^{\psi}$ and $r^{\phi}$ will be the last appearance of a copy of $r^{\psi}$ in $\varrho_{k+1}^{\phi}$.
Thus the proposition holds for $0 < k+1 \leq w$ as well.

The proof of the proposition above also shows that the outputs of the transitions of the two runs will be the same,
thus,
since the proposition holds for accepting runs as well,
$match(\varrho^{\psi})=match(\varrho^{\phi})=M$,
if $\varrho^{\psi}$ and $\varrho^{\psi}$ are accepting
(note that they must be either both accepting or both non-accepting).

One last touch is required.
The RMA $A^{\phi}$,
as explained,
can have runs of finite length.
On the other hand,
the original expression applies to (possibly infinite) streams.
Therefore, one last modification to $A^{\phi}$ is needed.
We add a loop, \trueb\ transition from the start state to itself,
so that a run may start at any point in the stream.
The ``effective'' maximum length of every run, however, remains $w$. 
The final RMA will then have the form of a tree 
(no cycles exist and walks can only split but not converge back again),
except for its start state with its self-loop.

We also note that $w$ must be a number greater than (or equal to) 
the minimum length of the walks induced by the accepting runs
(which is something that can be computed by the structure of the expression).
Although this is not a formal requirement,
if it is not satisfied,
then the RMA won't detect any matches.

\end{proof}

\end{document}